\documentclass{elsarticle}

\pdfoutput=1

\usepackage{latexsym}

\usepackage{amsmath}
\usepackage{amssymb}
\usepackage{amsfonts}
\usepackage{relsize} 

\renewcommand{\cite}[1]{\citep{#1}}

\usepackage{hyperref}

\newtheorem{theorem}{Theorem}
\newtheorem{proposition}{Proposition}
\newtheorem{lemma}{Lemma}

\newtheorem{definition}{Definition}
\newtheorem{example}{Example}

\newenvironment{proof}{\noindent{\em Proof}.\ }{\vskip\parskip\noindent}

\usepackage{ifthen}

\long\def\comment#1{}

\newcommand{\caI}{{\cal I}}
\newcommand{\caL}{{\cal L}}
\newcommand{\caP}{{\cal P}}
\newcommand{\caR}{{\cal R}}
\newcommand{\caS}{{\cal S}}
\newcommand{\caT}{{\cal T}}
\newcommand{\caX}{{\cal X}}

\newcommand{\sort}[1]{\ensuremath{\mathsf{#1}}}

\newcommand{\occ}[1]{\mathit{Pos}(#1)}
\newcommand{\funocc}[1]{\mathit{Pos}_{\Symbols}(#1)}

\newcommand{\dom}{{\cal D}om}
\newcommand{\ran}{{\cal R}an}
\newcommand{\var}{{\cal V}ar}

\newcommand{\Variables}{\caX}
\newcommand{\Symbols}{\Sigma}

\newcommand{\subterm}[2]{#1|_{#2}} 
\newcommand{\replace}[3]{#1[#3]_{#2}}

\newcommand{\idsubst}{\textit{id}}

\newcommand{\TermsOn}[5]{{\caT^{#4}_{#1}(#2)}{}_{#3}^{#5}}
\newcommand{\Terms}{\TermsOn{\Symbols}{\Variables}{}{}{}}
\newcommand{\ETerms}{\TermsOn{\Symbols\!/\!E}{\Variables}{}{}{}}
\newcommand{\TermsS}[1]{\TermsOn{\Symbols}{\Variables}{\sort{#1}}{}{}}
\newcommand{\ETermsS}[1]{\TermsOn{\Symbols\!/\!E}{\Variables}{\sort{#1}}{}{}}

\newcommand{\GTermsOn}[2]{\caT^{#2}_{#1}}
\newcommand{\GTerms}{\GTermsOn{\Symbols}{}}

\newcommand{\GTermsS}[1]{\GTermsOn{\Symbols,\sort{#1}}{}}

\newcommand{\TermsSt}{\TermsS{State}}

\newcommand{\composeRel}{;}
\newcommand{\rewrite}[1]{\rightarrow_{#1}}
\newcommand{\rewrites}[1]{\rightarrow^*_{#1}}
\newcommand{\rewritePos}[2]{\stackrel{#2}{\rightarrow}_{#1}}

\newcommand{\narrow}[2]{\stackrel{#1}{\rightsquigarrow}_{#2}}

\newcommand{\csu}[3]{\textit{CSU}_{#3}({#1})}
\newcommand{\restrict}[1]{{|_{#1}}}

\newcommand{\tuple}[1]{\ensuremath{\langle}\ensuremath{#1}\ensuremath{\rangle}}
\newcommand{\congr}[1]{=_{#1}}

\newcommand{\nI}[1]{\ensuremath{#1{\notin}\caI}}
\newcommand{\inI}[1]{\ensuremath{#1{\in}\caI}}

\newcommand{\caRP}{R_{\caP}}

\newcommand{\narrto}{\leadsto}

\newcommand{\sendto}{\hookrightarrow}

\title{State Space Reduction in\\ the Maude-NRL Protocol Analyzer
}

\author[escobar]{
Santiago Escobar}
\ead{sescobar@dsic.upv.es}
\address[escobar]{
DSIC-ELP, Universidad Polit\'ecnica de Valencia, 
Valencia, Spain
}

\author[meadows]{
Catherine Meadows}
\ead{meadows@itd.nrl.navy.mil}
\address[meadows]{
Naval Research Laboratory, 
Washington, DC, USA
}

\author[meseguer]{
Jos\'e Meseguer}
\ead{meseguer@cs.uiuc.edu}
\address[meseguer]{
University of Illinois at Urbana-Champaign, 
Urbana, IL, USA
}

\begin{document}

\pagestyle{plain}
\markright{}



\pagestyle{plain}


\begin{abstract}
The Maude-NRL Protocol Analyzer (Maude-NPA) is a tool and inference system for reasoning about the security of cryptographic protocols in which the cryptosystems satisfy different equational properties.  
It both extends and provides a formal framework for  
the original NRL Protocol Analyzer, 
which supported equational reasoning in a more limited way.
Maude-NPA
supports a wide variety of
 algebraic properties that includes many crypto-systems of interest
 such as, for example, one-time pads and Diffie-Hellman.
Maude-NPA, like the original NPA, looks for attacks by searching backwards from an insecure attack state, and assumes an unbounded number of sessions.  Because of the unbounded number of sessions and the support for different equational theories, it is necessary to develop ways of reducing 
the search space and avoiding infinite search paths.  
In order for the techniques to prove useful, they need not only to speed up the search, but should not violate completeness, so that failure to find attacks still guarantees security.  
In this paper we describe some state space reduction techniques that we have implemented in Maude-NPA.  
We also provide completeness proofs, and experimental evaluations of their 
effect on the performance of Maude-NPA.
\end{abstract}

\maketitle

\section{Introduction}

The Maude-NPA \cite{EscMeaMes-tcs06,FOSAD07} is a tool and inference system for reasoning about the security
of cryptographic protocols in which the cryptosystems satisfy different equational properties.
The tool handles searches in the unbounded session model, and thus can be 
used to provide proofs of security as well as to search for attacks.
It is the next generation of the NRL Protocol Analyzer
\cite{meadows-NRL94}, a tool that supported limited equational reasoning
and was successfully applied to the analysis of many different protocols.  
In Maude-NPA we improve on the original NPA in three ways.  
First of all, unlike NPA, which required considerable interaction with the user, Maude-NPA
is
completely automated (see \cite{FOSAD07}).
Secondly, its inference system has a formal basis in terms of rewriting logic and narrowing, which
allows us to provide proofs of soundness and completeness 
(see \cite{EscMeaMes-tcs06}).
Finally, the 
tool's inference system
supports reasoning modulo
the algebraic properties of
cryptographic and other functions
(see 
\cite{escobar-meadows-meseguer-secret06,escobar-hendrix-meadows-meseguer-secret07,STM10}).
Such algebraic properties are 
expressed as 
equational theories $E=E' \uplus Ax$ whose equations $E'$ are
confluent, coherent, and 
terminating rewrite rules modulo 
equational axioms $Ax$ such as commutativity ($C$), 
associativity-commutativity ($AC$),
or associativity-commutativity plus identity ($ACU$)
of some function symbols.
The Maude-NPA has then both dedicated and generic methods 
for solving unification problems in such theories $E=E' \uplus Ax$
\cite{escobar-meseguer-sasse-wrla08,Escobar-WRLA10,Escobar-JLAP},
which under appropriate checkable conditions 
\cite{EscobarMeseguerSasseRTA08}
yield finitary unification
algorithms.

Since Maude-NPA allows reasoning in the unbounded session model, and because it allows
reasoning about different equational theories 
(which typically generate many more solutions to unification problems
than syntactic unification, leading to bigger state spaces), 
it is necessary to find ways of pruning the search space in order to prevent 
infinite or overwhelmingly large search spaces.  
One technique for preventing 
infinite searches is  the generation of formal grammars describing terms 
unreachable by the intruder
(see \cite{meadows-NRL94,EscMeaMes-tcs06}
and Section~\ref{sec:grammars}).  However, grammars do not prune out all infinite searches, since unbounded session security is undecidable,
and there is a need for
other techniques.  Moreover, even when a search space is finite it may still be necessary to reduce it to a manageable size, and 
state space reduction
techniques for doing that will be necessary.  In this paper
we describe some of the major 
state space reduction
techniques that we have implemented in Maude-NPA, and provide completeness proofs and experimental evaluations
demonstrating an average state-space size reduction of
$99\%$ (i.e., the average size of the reduced state space is $1\%$ of that of
the original one) in the examples we have evaluated.
Furthermore, we show our combined techniques effective in obtaining
a \emph{finite} state space for all protocols in our experiments.

The optimizations we describe in this paper were designed specifically for Maude-NPA, and work within
the context of Maude-NPA search techniques.  However, although different tools use different models and search algorithms,
they all have a commonality in their syntax and semantics that means that, with some adaptations, optimization
techniques developed for one tool or type of tools can be applied to different tools as well.  Indeed, we have
already seen such common techniques arise, for example the technique of giving priority to input or output messages
respectively when backwards or forwards search is used (used by us and by Shmatikov and Stern in [20]) and the
use of the lazy intruder (used by us and, in a different form, by the On-the-Fly Model Checker [1]).  One of our motivations
of publishing our work on optimizations is to encourage the further interaction and adaptation of the techniques for use
in different tools.


The rest of the paper is organized as follows.  
After some preliminaries in Section \ref{sec:preliminaries},
we describe in Section \ref{sec:maude-npa} 
the model of computation used by the Maude-NPA. 
In Section
\ref{sec:optimizations},
 we 
describe 
the various state space reduction techniques 
that have been introduced to control state explosion, 
and give proofs of their completeness
as well as
showing their relations to other optimization techniques in the literature.  
We first 
briefly describe how automatically generated grammars
provide the main reduction that
cuts down the search space.
Then, we describe
how we 
obtain a second important state-space reduction 
by reducing the number of logical variables present in a state.
The additional state space reduction techniques presented in this paper are:
(i) giving priority to input messages in strands,
(ii) early detection of inconsistent states (that will never reach an initial state),
(iii) a relation of transition subsumption (to discard transitions and states
already being processed in another part of the search space),
and
(iv) the super-lazy intruder 
(to delay the generation of substitution instances as much as possible).
In Section \ref{sec:experiments} we describe 
our experimental evaluation of these state-space reduction techniques.  
In Section \ref{sec:conclusions} we describe future
work and conclude the paper.
This is an extended and improved version of 
\cite{EscobarMeadowsMeseguerESORICS08},
including proofs of all the results,
a refinement of the interaction between the transition subsumption
and the super-lazy intruder (Section~\ref{sec:interaction}), 
more examples and explanations,
as well as more benchmarked protocols.

%
\section{Background on Term Rewriting}
\label{sec:preliminaries}

We follow the classical notation and terminology from
\cite{Terese03} for term rewriting
and from
\cite{Meseguer92-tcs,tarquinia} for rewriting logic and order-sorted notions.
We assume an \textit{order-sorted signature} $\Symbols$
with a finite poset of sorts $(\sort{S},\leq)$ and a finite 
number of function symbols.
We assume an $\sort{S}$-sorted family 
$\Variables=\{\Variables_\sort{s}\}_{\sort{s} \in \sort{S}}$
of mutually disjoint variable sets with each $\Variables_\sort{s}$
countably infinite.
$\TermsS{s}$
denotes the set of terms of sort \sort{s},
and
$\GTermsS{s}$ the set of ground terms of sort \sort{s}.
We write 
$\Terms$ and $\GTerms$ for the corresponding term algebras.
We write $\var(t)$ for the set of variables present in a term $t$.
The set of positions of a term $t$ is written $\occ{t}$,
and
the set of non-variable positions $\funocc{t}$.
The subterm of $t$ 
at position $p$
is $\subterm{t}{p}$, and $\replace{t}{p}{u}$ is 
the result of replacing $\subterm{t}{p}$ by $u$ in $t$.
A \textit{substitution} $\sigma$ is a sort-preserving mapping
from a finite subset of $\Variables$, written $\dom(\sigma)$,
to $\Terms$.
The set of variables
introduced by $\sigma$ is $\ran(\sigma)$.
The identity
substitution is $\idsubst$.
Substitutions are homomorphically extended to $\Terms$.
The restriction of $\sigma$ to a set of variables $V$ is 
${\sigma}\restrict{V}$.
The composition of two substitutions is 
$(\sigma\circ\theta)(X)=\theta(\sigma(X))$ for $X\in\Variables$.

A \textit{$\Symbols$-equation} is an unoriented pair $t = t'$,
where $t \in \TermsS{s}$,
$t' \in \TermsS{s'}$,
and $s$ and $s'$ are sorts in the same connected
component of the poset  $(\sort{S},\leq)$.
Given a set $E$ of $\Symbols$-equations,
order-sorted equational logic
induces 
a congruence relation $\congr{E}$ on terms $t,t' \in \Terms$
(see \cite{tarquinia}).   Throughout this
paper we assume that $\GTermsS{s}\neq\emptyset$ for every sort \sort{s}.
We denote the $E$-equivalence class of a term $t\in\Terms$
as $[t]_E$ and the $E$-equivalence classes of all terms $\Terms$
and $\TermsS{s}$
as $\ETerms$ and $\ETermsS{s}$, respectively.

For a set $E$ of $\Symbols$-equations, an \textit{$E$-unifier} for a $\Symbols$-equation $t = t'$ is a
substitution $\sigma$ s.t. $\sigma(t) \congr{E} \sigma(t')$. 
A \textit{complete} set of $E$-unifiers of an
equation $t = t'$ is written
$\csu{t = t'}{W}{E}$. 
We say $\csu{t = t'}{W}{E}$ is \emph{finitary} if it contains a finite
number of $E$-unifiers.
$\csu{t = t'}{W}{}$
denotes a complete set of 
syntactic order-sorted unifiers between terms $t$ and $t'$,
i.e., without any equational property.

A \textit{rewrite rule} is an oriented pair $l \to r$, where
$l \not\in \Variables$
and
$l,r \in \TermsS{s}$ for some sort $\sort{s}\in\sort{S}$. 
An \textit{(unconditional)
  order-sorted rewrite theory} is a triple $\caR = (\Symbols,E,R)$
with $\Symbols$ an order-sorted signature, $E$ a set of
$\Symbols$-equations, and $R$ a set of rewrite rules.  
A \emph{topmost rewrite theory} $(\Symbols,E,R)$ is a rewrite theory 
s.t.
for each $l \to r \in R$, $l,r\in\TermsSt$ for a top sort \sort{State},
$r \not\in \Variables$,
and no operator in $\Symbols$ has \sort{State} as an argument sort.

The rewriting relation $\rewrite{R}$ on
$\Terms$ is 
$t \rewritePos{R}{p} t'$ 
(or $\rewrite{R}$) 
if 
$p \in \funocc{t}$,
$l \to r\in R$, 
$\subterm{t}{p} = \sigma(l)$, 
and $t' =
\replace{t}{p}{\sigma(r)}$
for some $\sigma$.
The relation $\rewrite{R/E}$
on $\Terms$ is 
$\congr{E} \composeRel\rewrite{R}\composeRel\congr{E}$,
i.e.,
$t \rewrite{R/E} s$
iff
$\exists u_1,u_2\in\Terms$ s.t.
$t \congr{E} u_1$,
$u_1 \rewrite{R} u_2$,
and
$u_2 \congr{E} s$.
Note that
$\rewrite{R/E}$ on $\Terms$
induces a relation 
$\rewrite{R/E}$ on $\ETerms$
by
$[t]_{E} \rewrite{R/E} [t']_{E}$ iff $t \rewrite{R/E} t'$.  

When $\caR = (\Symbols,E,R)$ is a topmost rewrite theory,
we can safely restrict ourselves to 
the general rewriting relation $\rewrite{R,E}$ on
$\Terms$, where
the rewriting relation $\rewrite{R,E}$ on
$\Terms$ is 
$t \rewritePos{R,E}{p} t'$ 
(or $\rewrite{R,E}$) 
if 
$p \in \funocc{t}$,
$l \to r\in R$, 
$\subterm{t}{p} \congr{E} \sigma(l)$, 
and $t' =
\replace{t}{p}{\sigma(r)}$
for some $\sigma$.
Note that
$\rewrite{R,E}$ on $\Terms$
induces 
a relation
$\rewrite{R,E}$ on $\ETerms$
by
$[t]_{E} \rewrite{R,E} [t']_{E}$ iff 
$\exists w\in\Terms$ s.t. $t \rewrite{R,E} w$ and $w \congr{E} t'$.  
We say that a term $t$ is \emph{$R,E$-irreducible} 
if there is no term $t'$ such that $t \rewrite{R,E} t'$;
this is extended to substitutions in the obvious way.

The narrowing relation $\narrow{}{R}$ on
$\Terms$ is 
$t \narrow{p}{\sigma,R} t'$ 
(or $\narrow{}{\sigma,R}$, $\narrow{}{R}$) 
if 
$p \in \funocc{t}$,
$l \to r\in R$, 
$\sigma \in \csu{\subterm{t}{p} = l}{W}{}$, 
and $t' = \sigma(\replace{t}{p}{r})$.  
Assuming that $E$ has a finitary and complete unification algorithm,
the narrowing relation $\narrow{}{R,E}$ on
$\Terms$ is
$t \narrow{p}{\sigma,R,E} t'$ 
(or $\narrow{}{\sigma,R,E}$, $\narrow{}{R,E}$) 
if 
$p \in \funocc{t}$,
$l \to r\in R$, 
$\sigma\in\csu{t|_p = l}{V}{E}$, 
and $t' = \sigma(\replace{t}{p}{r})$.

The use of topmost rewrite theories 
is entirely natural for communication protocols, since all state transitions can be viewed as changes of the global distributed state.
It also 
provides several advantages (see \cite{narrowing-hosc06}):
(i) as pointed out above the relation 
$\rewrite{R,E}$ achieves the same effect as the relation $\rewrite{R/E}$,
and
(ii) 
we obtain a completeness result between narrowing ($\narrow{}{R,E}$) and 
rewriting ($\rewrite{R/E}$). 


\begin{theorem}[Topmost Completeness]{\rm\cite{narrowing-hosc06}}\label{thm:hosc06}
Let $\caR=(\Symbols,E,R)$ be a topmost rewrite theory,
$t,t'\in\Terms$, and 
let $\sigma$ be a substitution such that
$\sigma(t) \to^*_{R,E} t'$.
Then, there are substitutions $\theta,\tau$ 
and a term $t''$ such that 
$t \narrto^*_{\theta,R,E} t''$,
$\sigma(t)\congr{E}\tau(\theta(t))$, and $t'\congr{E}\tau(t'')$.
\end{theorem}

In this paper, we consider only equational theories
$E=E' \uplus Ax$ such that the rewrite rules $E'$ are
confluent, coherent, and 
terminating modulo 
axioms $Ax$ such as commutativity ($C$), 
associativity-commutativity ($AC$),
or associativity-commutativity plus identity ($ACU$)
of some function symbols.
We also require axioms $Ax$ to be regular, i.e.,
for each equation $l = r \in Ax$, $\var(l) = \var(r)$.
Note that axioms such as
commutativity ($C$), 
associativity-commutativity ($AC$),
or associativity-commutativity plus identity ($ACU$) are regular.
The Maude-NPA has then both dedicated and generic methods 
for solving unification problems in such theories $E' \uplus Ax$
\cite{escobar-meseguer-sasse-wrla08,Escobar-WRLA10,Escobar-JLAP}.


\section{Maude-NPA's Execution Model}\label{sec:maude-npa}

Given a protocol $\mathcal{P}$, we first explain how its states
are modeled algebraically.  The key idea is to model protocol states as
elements of an initial algebra $T_{\Sigma_{\mathcal{P}}/
  E_{\mathcal{P}}}$, where $\Sigma_{\mathcal{P}}$ is the signature
defining the sorts and function symbols for the cryptographic functions and for
all the state constructor symbols,
and
$E_{\mathcal{P}}$ is a set of equations specifying the
\emph{algebraic properties} of the cryptographic functions and the state constructors.
Therefore, a state is an $E_{\mathcal{P}}$-equivalence class $[t]\in
T_{\Sigma_{\mathcal{P}}/ E_{\mathcal{P}}}$ with $t$ a ground
$\Sigma_{\mathcal{P}}$-term.  
However,
since the number of 
states $T_{\Sigma_{\mathcal{P}}/ E_{\mathcal{P}}}$
is in general infinite, rather than exploring concrete protocol states
$[t]\in T_{\Sigma_{\mathcal{P}}/ E_{\mathcal{P}}}$ we 
explore
\emph{symbolic state patterns} $[t(x_{1},\ldots,x_{n})] \in
T_{\Sigma_{\mathcal{P}}/ E_{\mathcal{P}}}(\Variables)$ on the free
$(\Sigma_{\mathcal{P}},E_{\mathcal{P}})$-algebra over a set of
variables $\Variables$.  
In this way, a state pattern $[t(x_{1},\ldots,x_{n})]$
represents not a single concrete state but a possibly infinite set of
such states, namely all the instances of the pattern
$[t(x_{1},\ldots,x_{n})]$ where the variables $x_{1},\ldots,x_{n}$
have been instantiated by concrete ground terms.

In the Maude-NPA \cite{EscMeaMes-tcs06,FOSAD07}, a \emph{state} in the protocol execution 
is a term $t$ of sort \sort{State},
$t \in T_{\Sigma_{\cal P}/ E_{\cal P}}(X)_{\sort{State}}$. 
A state is then a multiset built by an associative and commutative union
operator $\_\&\_$ with identity operator $\emptyset$.  Each element in
the multiset is either a strand or the intruder's knowledge at that state,
both explained below.

A \emph{strand}
\cite{strands} represents the sequence of messages sent and received
by a principal executing the protocol or by the intruder.
A principal sending (resp. receiving) a message $msg$ is represented by 
$msg^+$ (resp. $msg^-$).
We write ${m}^{\pm}$ to denote $m^{+}$ or $m^{-}$,
indistinctively.
We often write $+(m)$ and $-(m)$ instead of $m^{+}$ and $m^{-}$, respectively.
A strand 
is then a list 
$[msg_1^\pm,\ msg_2^\pm,\ msg_3^\pm,\linebreak[2]
  \ldots,\ msg_{k-1}^\pm,\ msg_k^\pm]$
describing the sequence of send and receive actions of a principal role in a protocol,   
  where each $msg_{i}$ is a term of
a special sort \textsf{Msg} described below,
i.e., $msg_{i}\in T_{\Sigma_{\mathcal{P}}/
  E_{\mathcal{P}}}(X)_{\textsf{Msg}}$. 
In Maude-NPA, strands evolve over time
as the send and receive actions take place,
and thus we use the symbol $|$ to divide past and future in a strand,
i.e.,
$[nil, \linebreak[2]msg_1^\pm, \linebreak[2]\ldots, \linebreak[2]msg_{j-1}^\pm \linebreak[2]\mid \linebreak[2]msg_j^\pm, \linebreak[2]msg_{j+1}^\pm, \linebreak[2]\ldots, \linebreak[2]msg_k^\pm, \linebreak[2]nil  ]$
where $msg_1^\pm,\linebreak[2] \ldots,\linebreak[2] msg_{j-1}^\pm$ are
the past messages, and $msg_{j}^\pm, msg_{j+1}^\pm, \ldots, msg_k^\pm$
are the future messages ($msg_{j}^\pm$ is the immediate future
message).
The nils are present so that 
the bar may be placed at the beginning or end of the strand if necessary.
A strand 
$[msg_1^\pm, \linebreak[2]\ldots, \linebreak[2] msg_k^\pm]$ 
is a shorthand for 
$[nil ~|~ msg_1^\pm,\linebreak[2] \ldots,\linebreak[2] msg_k^\pm , nil ]$.
We often remove the nils for clarity, except when there is nothing
else between the vertical bar and the beginning or end of a strand.
We write ${\cal S}_{\cal P}$ for the set of strands in 
the specification of the protocol $\caP$, including
the strands that describe the intruder's behavior.

The \emph{intruder's knowledge} is represented as a multiset of facts unioned
together with an associative and commutative union operator
\verb!_,_!  with identity operator $\emptyset$.  There are two kinds
of intruder facts: positive knowledge facts (the intruder knows message expression $m$,
i.e., $\inI{m}$), and negative knowledge facts (the intruder \emph{does not
yet know} $m$ but \emph{will know it in a future state}, i.e., $\nI{m}$).

Maude-NPA uses a special sort \sort{Msg} of messages
that allows the protocol specifier to describe 
other sorts 
as subsorts of the top sort \sort{Msg}.
The specifier can make use of another special sort $\sort{Fresh}$
in the protocol-specific signature $\Symbols$
for representing fresh unguessable values,
e.g., nonces.
The meaning of a variable of sort 
$\sort{Fresh}$ is that it will never be instantiated
by an $E$-unifier generated during the 
protocol
analysis.
This ensures that if two nonces are represented using 
different variables of sort 
$\sort{Fresh}$, they will never be identified 
and no approximation for nonces is necessary. 
We make 
explicit
the $\sort{Fresh}$ variables 
$r_1,\ldots,r_k (k \geq 0)$
generated by a strand 
by writing 
${:: r_1,\ldots,r_k::}\ [msg_1^\pm,\ldots,msg_n^\pm]$,
where 
each $r_i$ appears first in an output message $msg_{j_i}^+$ and can later be used
in any input and output message of $msg_{j_i+1}^\pm,\ldots,msg_n^\pm$.
Fresh variables generated by a strand are unique to that strand.

Let us introduce 
the well-known Diffie-Hellman protocol as a motivating example.
\begin{example}\label{ExProtocolStrands}\label{ExDH}
The Diffie-Hellman protocol 
uses exponentiation 
to share a secret between two parties, Alice and Bob.
There is a public constant, denoted by $g$, which will be the base of 
the exponentiations.
We represent the product of exponents by using the symbol $*$. 
Nonces 
are represented by $N_X$, denoting a nonce created by principal $X$.
Raising message $M$ to the power of exponent $X$ 
is denoted by $(M)^{X}$.
Encryption of message $M$ using the key $K$ is denoted by $\{M\}_K$.
The protocol description is as follows. 
\begin{enumerate}
\item 
$A \sendto B: \{ A\ ;\ B\ ;\ g^{N_{A}}\}$\\
Alice sends her name,
Bob's name,
and an exponentiation of a new nonce $N_{A}$ created by her 
to Bob.
\item
$B \sendto A: \{ A\ ;\ B\ ;\ g^{N_{B}}\}$\\
Bob sends his name,
Alice's name,
and an exponentiation of a new nonce $N_{B}$ created by him
to Alice.
\item
$A \sendto B : \{secret\}_{{g^{N_{A}}}^{N_B}}$\\
Bob receives $g^{N_{A}}$ and he raises it to the $N_B$ to obtain
the key 
${g^{N_A}}^{N_B}$. He sends a secret to Alice encrypted using the key.
 Likewise, when 
Alice receives  $g^{N_{B}}$, she raises it to the $N_A$, to obtain the key
${g^{N_B}}^{N_A}$.
We assume 
that exponentiation satisfies the equation
$ {g^{N_A}}^{N_B} = g^{N_A * N_B}$
and that the product operation \verb!_*_! is associative and commutative, so that
$${g^{N_B}}^{N_A} 
= {g^{N_A}}^{N_B}
= g^{N_B * N_A} $$ 
and 
therefore 
both 
Alice and Bob
share the same key.
\end{enumerate}

In the Maude-NPA's formalization of the protocol, 
we explicitly specify the signature $\Symbols$ describing 
the sorts and operations for 
messages, nonces, etc.
A nonce $N_A$ is denoted by $n(A,r)$, 
where $r$ is a unique variable of sort $\sort{Fresh}$.
Concatenation of two messages, e.g., $N_A$ and $N_B$, is denoted by the operator
$\_{;}\_$, e.g., $n(A,r)\ ;\ n(B,r')$.
Encryption of a message $M$ 
is denoted by $e(A,M)$, e.g., 
$\{N_B\}_{K_B}$ is denoted by
$e(K_{B},n(B,r'))$.
Decryption is similarly denoted by $d(A,M)$.
Raising a message $M$ to the power of an exponent $E$ (i.e., $M^{E}$)
is denoted by $exp(M,E)$, e.g., 
$g^{N_{B}}$ is denoted by
$exp(g,n(B,r'))$.
Associative-commutative multiplication  of nonces
is denoted by $\_{*}\_$. 
A secret generated by a principal 
is denoted by $sec(A,r)$,
where $r$ is a unique variable of sort $\sort{Fresh}$.
The protocol-specific signature $\Symbols$ 
contains 
the following subsort relations 
$(\sort{Name}, \sort{Nonce}, \sort{Secret}, \sort{Enc}, \sort{Exp} < \sort{Msg})$ 
and
$(\sort{Gen}, \sort{Exp} < \sort{GenvExp})$
and the following operators:
%
\[
\begin{array}{c@{\ \ \ \ \ \ \ \ }c}
a, b, i : \ \rightarrow \sort{Name}
&
g : \rightarrow \sort{Gen} 
\\
n : \sort{Name} \times \sort{Fresh} \rightarrow \sort{Nonce}
&
sec : \sort{Name} \times \sort{Fresh} \rightarrow \sort{Secret}
\\
\_\,{;}\,\_\ : \sort{Msg} \times \sort{Msg} \rightarrow \sort{Msg}
&
e, d : \sort{Key} \times \sort{Msg} \rightarrow \sort{Enc}
\\
exp : \sort{GenvExp} \times \sort{Nonce} \rightarrow \sort{Exp}
&
\verb!_*_! : \sort{Nonce} \times \sort{Nonce} \rightarrow \sort{Nonce}
\end{array}
\]
\noindent
In the following we will use letters $A,B$ for variables of sort $\sort{Name}$,
letters $r,r',r''$ for variables of sort $\sort{Fresh}$,
and
letters $M,M_1,M_2,Z$ for variables of sort $\sort{Msg}$;
whereas letters $X,Y$ will also represent variables, but their sort
will depend on the concrete position in a term.
The encryption/decryption cancellation properties are described using 
the 
equations 
$$e(X,d(X,Z)) = Z \mbox{ and } d(X,e(X,Z)) = Z$$ 
in $E_{\caP}$.
The key algebraic property of exponentiation, $z^{x^y} = z^{x*y}$,
is described using 
the 
equation 
$$exp(exp(W,Y),Z) = exp(W,Y * Z)$$
in $E_{\caP}$
(where $W$ is of sort \sort{Gen} instead of the more general sort
\sort{GenvExp} in order to provide a finitary 
narrowing-based unification procedure
modulo $E_{\caP}$, see \cite{escobar-hendrix-meadows-meseguer-secret07}
for details on this concrete equational theory). 
Although
multiplication modulo a prime number has a unit and inverses, 
we have only included the algebraic properties that are necessary for Diffie-Hellman to work.
The two strands $\caP$ 
associated to the protocol roles, Alice and Bob, shown above 
are:%
%
$$
:: r,r' ::
[\ (A ; B ; exp(g,n(A,r)))^+,\ (B ; A ; X)^-,\  (e(exp(X,n(A,r)), sec(A,r')))^+]$$
$$
:: r'' ::
[\  (A ; B ; Y)^-,\ (B ; A ; exp(g,n(B,r'')))^+,\ 
(e(exp(Y,n(B,r'')),\textit{SR})^- ]$$

\noindent
The following strands describe the intruder abilities 
according
to the Dolev-Yao attacker's capabilities \cite{dolev-yao}.

\begin{itemize}
\item\label{s3a} $[M_1^-, M_2^-, (M_1 ; M_2)^+]$ Concatenation \\[-.3cm]
\item\label{s3b} $[(M_1 ; M_2)^-, M_1^+]$ Left-deconcatenation \\[-.3cm]
\item\label{s3c} $[(M_1 ; M_2)^-, M_2^+]$ Right-deconcatenation \\[-.3cm]
\item\label{s4e} $[\ K^-, M^-, e(K,M)^+ \ ]$ Encryption \\[-.3cm]
\item\label{s4d} $[\ K^-, M^-, d(K,M)^+ \ ]$ Decryption \\[-.3cm]
\item\label{s4} $[\ M_1^-, M_2^-, (M_1 * M_2)^+ \ ]$ Multiplication \\[-.3cm]
\item\label{s5} $[\  M_1^-, M_2^-, exp(M_1,M_2)^+  \ ]$ Exponentiation\\[-.3cm]
\item\label{s6} $[\   g^+  \ ]$ Generator \\[-.3cm]
\item\label{s7} $[  \   A^+  \ ]$ All names are public \\[-.3cm]
\item\label{s8} $::r'''::\;[\ n(i,r''')^+\  ]$ Generation of intruder nonces
\end{itemize}
Note that the intruder cannot extract information from either an exponentiation
or a product of exponents, but can only compose them. 
Also, the intruder cannot extract information directly
from an encryption but it can indirectly by using a decryption
and the cancellation of encryption and decryption,
which is an algebraic property,
i.e.,
$[K^-, e(K,M)^-, M^+ ] \congr{E_\caP} [ K^-, e(K,M)^-, d(K,e(K,m))^+  ]$.
\end{example}


\subsection{Backwards Reachability Analysis}\label{sec:rules}

Our protocol analysis methodology is then based on the idea of
\emph{backwards reachability analysis}, where we begin with one or more
state patterns corresponding to \emph{attack states}, and want to
prove or disprove that they are \emph{unreachable} from the set of
initial protocol states.  In order to perform such a reachability
analysis we must describe how states change as a consequence of
principals performing protocol steps and of intruder actions.
This can be done by describing such state changes by means of a set
$R_{\mathcal{P}}$ of \emph{rewrite rules}, so that the rewrite theory
$(\Sigma_{\mathcal{P}},E_{\mathcal{P}},R_{\mathcal{P}})$ 
characterizes the behavior of protocol $\mathcal{P}$ modulo the equations $E_{\mathcal{P}}$.  
In the case where new strands are not introduced into the state, 
the 
corresponding
rewrite rules in $R_\caP$ 
 are 
as follows\footnote{To simplify the exposition, we omit the fresh variables at the beginning of each strand in a rewrite rule.}, 
where 
$L,L_1,L_2$ denote lists of input and output messages
($+m$,$-m$),
$IK,IK'$ denote sets of intruder facts 
(\inI{m},\nI{m}),
and
$SS,SS'$ denote sets of strands:

\noindent
\begin{small}%
\begin{align}%
[L ~|~ M^-, L']\  \&\ SS\ \&\ (\inI{M},IK)
  &\to 
  [L, M^- ~|~ L']\  \&\ SS\ \&\ (\inI{M},IK)
  \label{eq:negative-1}\\
 [L ~|~ M^+, L']\  \&\ SS\ \&\ IK 
 \hspace{8.6ex}
  &\to 
  [L, M^+ ~|~ L']\  \&\ SS\ \&\ IK
  \label{eq:positiveNoLearn-2}\\
 [L ~|~ M^+, L']\  \&\ SS\ \&\ (\nI{M},IK) 
  &\to 
  [L, M^+ ~|~ L']\  \&\ SS\ \&\ (\inI{M},IK)
  \label{eq:positiveLearn-4}
\end{align}%
\end{small}%

In a \emph{forward execution} of the protocol strands,
Rule \eqref{eq:negative-1} 
describes a message reception event in which an input message
is received from the intruder; 
the intruder's knowledge acts
in fact as the only \emph{channel} through which all communication takes place.
Rule \eqref{eq:positiveNoLearn-2} 
describes a message send in which
the intruder's knowledge is not increased;
it is irrelevant where the message goes.
Rule \eqref{eq:positiveLearn-4} 
describes the alternative case of a send event such that
the intruder's knowledge is positively increased.
%
Note that
Rule \eqref{eq:positiveLearn-4} 
makes explicit \emph{when} the intruder learned
a message $M$, which was recorded in the previous state by the negative fact \nI{M}.
A fact $\nI{M}$ can be paraphrased as: 
``the intruder does not yet know $M$, but will learn it in the future''.
This enables a very important restriction of the tool, expressed by
saying that the intruder \emph{learns} a term \emph{only once} \cite{EscMeaMes-tcs06}:
if the intruder needs to use a term twice, then he must learn it the first time it is needed;
if he learns a term and needs to learn it again in a previous state, 
found later during the backwards search, then the state
will be discarded as unreachable.  
Note that 
Rules \eqref{eq:negative-1}--\eqref{eq:positiveLearn-4}
are \emph{generic}: they belong to $\caR_\caP$ for \emph{any}
protocol $\caP$.

It is also the case that when we are performing a backwards search, only the strands that we are searching for are listed explicitly: 
extra strands necessary to reach an initial state are dynamically added
to the state by explicit introduction through 
protocol-specific rewrite rules
(one for each output message $u^+$ in an honest or intruder strand in $\caS_\caP$)
as follows:

\noindent
\begin{small}%
\begin{align}%
\mbox{for each }[~ l_1,\ u^+,\ l_2~] \in \caS_{\caP}:
[~ l_1 ~|~ u^+, l_2~] \, \&\, SS\, \,\&\, (\nI{u},IK)
\to
SS \,\&\, (\inI{u},IK)
\label{eq:newstrand}%
\end{align}%
\end{small}%
\noindent
where 
$u$ denotes a message,
$l_1,l_2$ denote lists of input and output messages
($+m$,$-m$),
$IK$ denotes a set of intruder facts 
(\inI{m},\nI{m}),
and
$SS$ denotes a set of strands.
For example, intruder concatenation of two learned messages,
as well as the learning of such a concatenation by the intruder,
is described
as follows:

\noindent
\begin{small}%
\begin{align}
& [M_{1}^-, M_{2}^- ~|~ (M_{1} ; M_{2})^+ ]\  \&\ SS\ \&\ (\nI{(M_{1} ; M_{2})},IK) 
\to 
SS\ \&\ (\inI{(M_{1} ; M_{2})},IK)
  \nonumber
\end{align}%
\end{small}%

\noindent
This rewrite rule can be understood, in a backwards search, as
``in the current state
the intruder is able to learn a message that matches the pattern 
$M_{1} ; M_{2}$ 
if he is able to learn message $M_1$ and message $M_2$
in prior states".
In summary, 
for a protocol $\caP$,
the set $R_\caP$ of rewrite rules 
obtained from the protocol strands $\caS_\caP$
that are
used
for backwards narrowing reachability analysis
\emph{modulo} the equational properties $E_{\caP}$
is 
$R_{\caP} = \{ \eqref{eq:negative-1},\eqref{eq:positiveNoLearn-2},\eqref{eq:positiveLearn-4} \}
\cup\eqref{eq:newstrand}$.
These rewrite rules give the basic execution model of Maude-NPA.  However, as we shall see, it
will later be necessary to modify them in order to optimize the search.  In later sections of this paper
we will show how these rules can be modified to optimize the search while still maintaining completeness.

On the other hand, the assumption 
that algebraic properties are 
expressed as 
equational theories $E=E' \uplus Ax$ whose equations $E'$ are
confluent, coherent, and 
terminating rewrite rules modulo 
regular equational axioms $Ax$ such as commutativity ($C$), 
associativity-commutativity ($AC$),
or associativity-commutativity plus identity ($ACU$)
of some function symbols, 
implies some extra conditions on the rewrite theory $R_\caP$
(see \cite{EscMeaMes-tcs06}).
Namely, 
for any term $\inI{m}$ (resp. term $m^-$) and any 
$E'{,}Ax$-irreducible substitution $\sigma$,
$\inI{\sigma(m)}$ (resp. $(\sigma(m))^-$) must be 
$E'{,}Ax$-irreducible.
This is because many of our optimization techniques rely on the assumption that terms have a unique normal form modulo a regular
equational theory,
and achieve their results by reasoning about the normal forms of terms.

Finally, states have, in practice, another component containing 
the actual message exchange sequence between principal or intruder strands
(i..e, all the expressions $m^\pm$ exchanged between the honest and intruder strands). 
We do not make use of the message exchange sequence until 
Section~\ref{sec:interaction}, 
so we delay its introduction until there.


The way to
analyze \emph{backwards} reachability is then relatively easy, namely,
to run the protocol ``in reverse.''  This can be achieved by using the
set of rules $R^{-1}_{\mathcal{P}}$, where $v \longrightarrow u$ is in
$R^{-1}_{\mathcal{P}}$ iff $u \longrightarrow v$ is in
$R_{\mathcal{P}}$.  Reachability analysis can be performed
\emph{symbolically}, not on concrete states but on symbolic state
patterns $[t(x_{1},\ldots,x_{n})]_{E_\caP}$ by means of 
\emph{narrowing
  modulo} $E_{\mathcal{P}}$ (see Section~\ref{sec:preliminaries}).
We call \emph{attack patterns}
those states patterns (i.e., terms with logical variables) 
used to start the narrowing-based
backwards reachability analysis. 
An \emph{initial state} is a state where all strands have their vertical bar at the beginning
and there is no positive fact of 
the form $\inI{u}$ for a message term $u$ in the intruder's knowledge.
If no initial state is found during the backwards reachability analysis
from an attack pattern, the protocol has been proved secure
for that attack pattern with respect to the assumed intruder capabilities
and the algebraic properties.
If an initial state is found, then we conclude that
the attack pattern is possible
and
a concrete attack can be inferred from the exchange sequence stored in the initial state.
Note that an initial state may be generic, in the sense of having logical variables for those elements that are not relevant for the attack.

\begin{example}(Example \ref{ExDH} continued)\label{exDH:cont}
The attack pattern that we are looking for is one in which 
Bob completes the protocol and the intruder is able to learn the secret.  
The attack state pattern to be given as input to Maude-NPA is:%

{\small
$$
\begin{array}{@{}l@{\ }r@{}}
\begin{array}{@{}l@{}}
::r':: 
[\, (A ; B ; Y)^-, (B ; A ; exp(g,n(B,r')))^+, 
(e(exp(Y,n(B,r')),sec(a,r'')))^-\, |\, nil\, ]\\[1ex]
\&\ SS\ \&\ ( \inI{sec(a,r'')},\ IK)
\end{array}
&
(\dagger)
\end{array}
$$
}

\noindent
Using the above attack 
pattern
Maude-NPA 
is able to find an initial state of the protocol,
showing that the attack state is possible. 
Note that this initial state is generalized to two sessions in parallel:
one session where 
Alice (i.e., principal named $a$) is talking to
another principal $B'$
---in this session the intruder gets a nonce $n(a,r)$ originated from $a$---
and 
another session where
Bob (i.e., principal named $b$) 
is trying to talk to Alice.
If we instantiate $B'$ to be $b$, then one session is enough,
although the tool returns the most general attack.
The strands associated to the initial state 
found by the backwards search
are as follows:

{\small
$$
\begin{array}{@{}l@{}}
[nil\mid exp(g, n(a, r)))^-, Z^-, exp(g, Z * n(a, r))^+]\ \& \\[.5ex]
[nil\mid exp(g, Z * n(a, r))^-, e(exp(g, Z * n(a, r)), sec(a, r''))^-, sec(a, r'')^+]\ \& \\[.5ex]
[nil\mid exp(g, n(b, r')))^-, W^-, exp(g, W * n(b, r'))^+]\ \& \\[.5ex]
[nil\mid exp(g, W * n(b, r'))^-, sec(a, r'')^-, e(exp(g, W * n(b, r')), sec(a, r''))^+]\ \& \\[.5ex]
[nil\mid (a ; b ; exp(g, n(b, r')))^-, (b ; exp(g, n(b, r')))^+]\ \& \\[.5ex]
[nil\mid (b ; exp(g, n(b, r')))^-, exp(g, n(b, r'))^+]\ \& \\[.5ex]
[nil\mid (a ; B' ; exp(g, n(a, r)))^-, (B' ; exp(g, n(a, r)))^+]\ \& \\[.5ex] 
[nil\mid (B' ; exp(g, n(a, r)))^-, exp(g, n(a, r))^+]\ \& \\[.5ex]
::r'::\\{}
[nil \,{\mid}\, (a ; b ; exp(g, W))^-, (a ; b ; exp(g, n(b, r')))^+, e(exp(g, W*n(b, r')), sec(a, r''))^-]\, \&\\[.5ex]
::r'',r::\\{}
[nil\mid (a ; B' ; exp(g, n(a, r)))^+, (a ; B' ; exp(g, Z))^-, e(exp(g, Z * n(a, r)), sec(a, r''))^+]\ 
\end{array}
$$
}

\noindent
Note that the last two strands,
generating fresh variables $r,r',r''$, 
are protocol strands and the others are intruder strands.

The concrete 
message exchange sequence
obtained by the 
reachability analysis is the following:

{\small
$$
\begin{array}{@{}c@{\:}c@{\:}c@{}}
\begin{array}{@{}l@{}}
1. (a ; b ; exp(g, W))^- \\ 
2. (a ; b ; exp(g, n(b, r')))^+ \\
3.  (a ; b ; exp(g, n(b, r')))^- \\ 
4.  (b ; exp(g, n(b, r')))^+ \\
5. (b ; exp(g, n(b, r')))^- \\
6. (exp(g, n(b, r')))^+ \\  
7.  (exp(g, n(b, r')))^- \\
8. W^- \\
9. exp(g, W * n(b, r'))^+ \\
\end{array}
&
\begin{array}{@{}l@{}}
10.  (a ; B' ; exp(g, n(a, r)))^+ \\ 
11.  (a ; B' ; exp(g, n(a, r)))^-  \\ 
12. (B' ; exp(g, n(a, r)))^+\\
13. (B' ; exp(g, n(a, r)))^- \\ 
14.  (exp(g, n(a, r)))^+ \\ 
15.  (exp(g, n(a, r)))^- \\
16.  Z^- \\ 
17. exp(g, Z * n(a, r))^+ \\
\end{array}
&
\begin{array}{@{}l@{}}
18.  (a ; B' ; exp(g, Z))^- \\ 
19.  e(exp(g, Z * n(a, r)), sec(a, r''))^+  \\
20.  e(exp(g, Z * n(a,r)), sec(a, r''))^- \\ 
21.  exp(g, Z * n(a,r))^- \\ 
22. sec(a, r'')^+ \\
23.  exp(g, W * n(b, r'))^- \\
24.  sec(a, r'')^- \\ 
25. e(exp(g, W * n(b, r'),sec(a, r''))^+\\
26. e(exp(g, W * n(b, r')).sec(a,r''))^- 
\end{array}
\end{array}
$$
}

\noindent
Step 1) describes Bob (i.e., principal named $b$) receiving an initiating message from the intruder impersonating Alice.  Step 2) describes Bob sending the response, and Step 3) describes the intruder receiving it. 
Steps 4) through 9) describe the intruder computing the key 
$exp(g, W * n(b, r'))$
she will use to communicate with Bob.  Step 10) describes Alice initiating the protocol with a principal $B'$.  Step 11) describes the intruder receiving it, and steps 11) through 17) describe the intruder constructing the key 
$exp(g, Z * n(a, r))$
she will use to communicate with Alice.  Steps 18) and 19) describe Alice receiving the response from the intruder impersonating $B'$ and Alice sending the encrypted message.  Steps 20) through 22) describe the intruder decrypting the message to get the secret.  In steps 23) through 25) the intruder re-encrypts the secret with the key she shares with Bob and sends it, and in Step 26) Bob receives the message.

Note that
there are some intruder strands  missing in the initial state
because certain terms are assumed to be trivially generable by the intruder, and so not searched for;
namely, intruder strands generating 
variable $Z$,
variable $W$,
term
$(a ; b ; exp(g, W))$, and
term
$(a ; B' ; exp(g, Z))$.
Variables $Z$ and $W$ can be filled in with any nonce,
for instance nonces generated by the intruder,
such as
$W = n(i,r''')$ and $Z = n(i,r'''')$
in the following way:

{\small
$$
\begin{array}{l}
::r''':: [nil \mid (n(i,r'''))^+]\ \& 
::r'''':: [nil \mid (n(i,r''''))^+]\ 
\end{array}
$$
}

\noindent
Also, note that nonces $W$ and $Z$ are used by the intruder
to generate messages
$(a ; b ; exp(g, W))$ and
$(a ; B' ; exp(g, Z))$ 
in the following way:

{\small
$$
\begin{array}{l}
[nil \mid (a)^+]\ \&\ 
[nil \mid (b)^+]\ \&\ 
[nil \mid (B')^+]\ \&\ \\[.5ex]
[nil \mid (g)^+]\ \&\ 
[nil \mid (g)^-, W^-, exp(g,W)^+]\ \&\ 
[nil \mid (g)^-, Z^-, exp(g,Z)^+]\ \&\ \\[.5ex]
[nil \mid (a)^-, (b)^-, (a ; b)^+]\ \&\ 
[nil \mid (a ; b)^-, (exp(g,W))^-, (a ; b ; exp(g,W))^+]\ \& \\[.5ex]
[nil \mid (a)^-, (B')^-, (a ; B')^+]\ \&\ 
[nil \mid (a ; B')^-, (exp(g,Z))^-, (a ; B' ; exp(g,Z))^+]\ 
\end{array}
$$
}
\end{example}


\section{State Space Reduction Techniques}\label{sec:optimizations}

In this section we present Maude-NPA's state space reduction techniques.  
Before presenting them,
we formally identify two classes of states that can be safely removed: \emph{unreachable} and \emph{redundant} states.
We begin the presentation with the notion of grammars, and its associated state space reduction technique, which is the oldest Maude-NPA technique and does much to identify and remove non-terminating search paths.  In many cases (although not all) this is enough to turn 
an infinite search space into a finite one.  
We then describe a number of simple techniques which remove states that can be shown to be unreachable, thus saving the cost of searching for them.  We conclude by describing two powerful techniques for eliminating redundant states: subsumption partial order reduction and the super-lazy intruder, and we prove their completeness.

First, the Maude-NPA satisfies a very general completeness result.

\begin{theorem}[Completeness]{\rm\cite{EscMeaMes-tcs06}}\label{thm:EscMeaMes-tcs06}
Given a topmost rewrite theory $\caR_\caP = \linebreak[4](\Symbols_\caP,E_{\caP},R_{\caP})$
representing protocol $\caP$,
and a non-initial state $St$ (with logical variables),
if there is a substitution $\sigma$
and an initial state $St_{ini}$
such that
$\sigma(St) \to^*_{R_{\caP}^{-1},E_\caP} St_{ini}$,
then
there are substitutions $\sigma',\rho$ and an initial state $St'_{ini}$
such that
$St \narrto^*_{\sigma',R_{\caP}^{-1},E_\caP} St'_{ini}$,
$\sigma \congr{E_\caP} \sigma'\circ\rho$,
and
$St_{ini} \congr{E_\caP} \rho(St'_{ini})$.
\end{theorem}

Our optimizations are able to identify 
two kinds of unproductive states:
\emph{unreachable} and \emph{redundant} states.
\begin{definition}[Unreachable States]\label{def:unreachable}
Given a topmost rewrite theory $\caR_\caP = (\Symbols_\caP,E_{\caP},R_{\caP})$
representing protocol $\caP$,
a state $St$ (with logical variables) is \emph{unreachable}
if there is no 
sequence 
$St \narrto^*_{\sigma,R_{\caP}^{-1},E_\caP} St_{ini}$
leading to an initial state $St_{ini}$.
\end{definition}
\begin{definition}[Redundant States]
Given a topmost rewrite theory $\caR_\caP = (\Symbols_\caP,E_{\caP},R_{\caP})$
representing protocol $\caP$
and a state $St$ (with logical variables),
a backwards narrowing step 
$St \narrto_{\sigma_1,R_{\caP}^{-1},E_\caP} St_1$
is called \emph{redundant}
(or just state $St_1$ is identified as \emph{redundant})
if 
for any initial state $St_{ini1}$ reachable from $St_1$,
i.e., ${St_1 \narrto^*_{\theta_1,R_{\caP}^{-1},E_\caP} St_{ini1}}$,
there are states $St_2$ and $St_{ini2}$,
a narrowing step
$St \narrto_{\sigma_2,R_{\caP}^{-1},E_\caP} St_2$,
a narrowing sequence
$St_2 \narrto^*_{\theta_2,R_{\caP}^{-1},E_\caP} St_{ini2}$,
and a substitution $\rho$
such that 
$\sigma_1\circ\theta_1 \congr{E_\caP} \sigma_2\circ\theta_2\circ\rho$
and
$St_{ini1} \congr{E_\caP} \rho(St_{ini2})$.
\end{definition}
There are three reasons for wanting to detect
unproductive backwards narrowing reachability steps. 
One is to reduce, if possible, the initially infinite
search space to a finite one, as 
it is sometimes possible to do with the use of grammars, by removing unreachable states.
Another is to reduce the size
of a (possibly finite) search space by 
eliminating unreachable states early, 
i.e., 
before
they are eliminated by exhaustive search.
This elimination of unreachable states 
can have an effect far beyond eliminating a single node in the search space, 
since a single
unreachable state 
may appear multiple times and/or have multiple descendants. 
Finally, 
if there are several steps leading to the same initial state, as for redundant states,
then
it is also possible to use various partial order reduction techniques
that can further shrink the number of states that need to be explored.

\subsection{Grammars}\label{sec:grammars}

The   
Maude-NPA's ability to reason effectively 
about a protocol's algebraic properties is 
a result of 
its combination of symbolic reachability analysis using narrowing
modulo equational properties (see Section~\ref{sec:preliminaries}),
together with its grammar-based techniques  
for reducing the size of the search space.
The key idea of grammars is to detect terms
$t$ in positive facts $\inI{t}$
of the intruder's knowledge of a state $St$
that will never be transformed into a negative fact $\nI{\theta(t)}$
in any initial state $St'$ backwards reachable from $St$.
This means that $St$ can never reach an initial state and therefore
it 
can be safely discarded.
Here we briefly explain how grammars work 
as a state space reduction technique
and refer the reader to 
\cite{meadows96,EscMeaMes-tcs06}
for further details.
%
\emph{Automatically generated grammars }
$\tuple{G_1,\ldots,G_m}$
represent unreachability information (or co-invariants),
i.e., typically infinite sets of states unreachable 
from an initial state. 
These automatically generated grammars are very important in our framework,
since in the best case they 
can reduce the infinite search space to a finite one,  
or, at least, can drastically reduce the search space.

\begin{example}
Consider again the attack pattern $(\dagger)$ in Example~\ref{exDH:cont}. 
%
%
After a couple of backwards narrowing steps, the Maude-NPA finds
the following state:

{\small
$$
\begin{array}{l}%
[\ nil \mid (M ; sec(a,r''))^-,\ (sec(a,r''))^+\ ]\ \&\\[1ex]
::r':: 
[(A ; B ; Y)^-, (B ; A ; exp(g,n(B,r')))^+ \mid  
(e(exp(Y,n(B,r')),sec(a,r'')))^-\ ]\ \&\\[1ex]
(\ 
\inI{(M ; sec(a,r''))},\ 
\inI{e(exp(Y,n(B,r')),sec(a,r''))},\ 
\nI{sec(a,r'')}
\ )
\end{array}
$$
}

\noindent
which corresponds to the intruder obtaining (i.e., learning)
the message
$sec(a,r'')$ from a bigger message $(M; sec(a,r''))$, 
although the contents of variable $M$ have not yet been found
by the backwards reachability analysis.
This process of adding
more and more intruder strands that look for terms
$(M' ; M ; sec(a,r''))$
$(M'' ; M' ; M ; sec(a,r''))$, 
$\ldots$
can go on forever. 
Note that if we carefully check the strands
for the protocol,
we can see that the honest strands either 
never produce a message with normal form 
``$M ; secret$'' or 
such a message is under a public key encryption (and thus the intruder
cannot get the contents), 
so the previous state is clearly unreachable and can be 
discarded.
The grammar, which is generated by Maude-NPA, 
capturing the previous state as unreachable, is as follows:

{\small
\begin{verbatim}
grl M inL => e(K, M) inL . ; 
grl M inL => d(K, M) inL . ; 
grl M inL => (M ; M') inL . ; 
grl M inL => (M' ; M) inL . ; 
grl M notInI, 
    M notLeq exp(g, n(A, r)), 
    M notLeq B ; exp(g, n(A, r')) => (M' ; M) inL .)
\end{verbatim}
}

\noindent
where all the productions and exceptions refer to normal forms of messages
w.r.t. the equational theory $E_\caP$.

Intuitively, the last production rule in the grammar above says
that any term with normal form $M' ; M$
cannot be learned by the intruder if the subterm $M$ is different from 
$exp(g, n(A, r))$ and $B ; exp(g, n(A, r'))$
(i.e., it does not match such patterns)
and the constraint \nI{M} appears explicitly 
in the intruder's knowledge of the current state being checked
for unreachability.
Moreover, any term of any of the following normal forms: 
$e(A, M)$,
$d(A, M)$, 
$(M' ; M)$,
or $(M ; M')$ cannot be learned by the intruder
if subterm $M$ is also not learnable by the intruder.
\end{example}


\subsection{Public data}\label{sec:public}

The simplest optimization possible 
is one that can be provided explicitly by the user.
When we are searching for some data
that we know is easy to learn by the intruder,
the tool can avoid this by assuming 
that such data is \emph{public}.
Such data is considered public 
by using a special sort \sort{Public}
and a subsort definition, e.g.
``\texttt{subsort Name < Public}''.
That is, 
given a state $St$ that
contains an expression \inI{t} in the intruder's knowledge
where $t$ is of sort \texttt{Public}, we can remove the expression 
\inI{t} from the intruder's knowledge, since the backwards reachability steps
taken care of such a \inI{t} are necessary in order to lead to an initial state
but their inclusion in the message sequence is superfluous.
The completeness proof for this optimization is trivial and thus omitted.

\subsection{Limiting Dynamic Introduction of New Strands}\label{sec:limiting}


As pointed out in Section~\ref{sec:rules}, rules of type \eqref{eq:newstrand} 
allow the dynamic introduction of new strands.
However, new strands can also be introduced by unification
of a state containing a variable $SS$ denoting a set of strands
and one of the rules of \eqref{eq:negative-1}, \eqref{eq:positiveNoLearn-2}, and \eqref{eq:positiveLearn-4}, where 
variables $L$ and $L'$ denoting lists of input/output messages
will be introduced by instantiation of $SS$.
The same can happen with new intruder facts of the form
$\inI{X}$, where $X$ is a variable, by instantiation of a variable $IK$
denoting the rest of the intruder knowledge.

\begin{example}
Consider a state $St$ 
of the form $SS \,\&\, IK$
where 
$SS$ denotes a set of strands
and 
$IK$ denotes a set of facts in the intruder's knowledge.
Now, consider
Rule \eqref{eq:negative-1}:
$$SS'\,\&\,[L ~|~ M^-, L']\,\&\,(\inI{M},IK') \to SS'\,\&\,[L, M^- ~|~ L']\,\&\,(\inI{M},IK')$$
\noindent
The following backwards narrowing step applying such a rule
can be performed from 
$St = SS \,\&\, IK$ 
using the unifier
$\sigma=\{SS \mapsto SS'\,\&\,[L, M^- ~|~ L'], IK \mapsto (\inI{M},IK')\}$
$$
SS\ \&\ IK
\narrow{\sigma}{R,E}
SS'\,\&\,[L ~|~ M^-, L']\,\&\,(\inI{M},IK')
$$
but this backwards narrowing step is unproductive, since it is not guided by the information in the attack state.
Indeed, the same rule can be applied again
using variables $SS'$ and $IK'$
and this can be repeated many times.
\end{example}

In order to avoid a huge number of unproductive narrowing steps
by useless instantiation,
we allow the introduction of new strands and/or new intruder facts 
\emph{only by rule application} instead of just by unification. 
For this, we do two things:
\begin{enumerate}
\item we remove any of the following variables from attack patterns: 
$SS$ denoting a set of strands, $IK$ denoting a set of intruder facts, and $L,L'$ denoting a set of input/output messages;
and 
\item we replace Rule
 \eqref{eq:negative-1} by the following Rule \eqref{eq:negative:back},
since we do no longer have a variable denoting a set of intruder facts
that has to be instantiated:
\begin{align}
&SS\,\&\,[L ~|~ M^-, L']\,\&\,(\inI{M},IK) \,{\to}\, SS\,\&\,[L, M^- ~|~ L']\,\&\,IK 
  \label{eq:negative:back}
\end{align}%
\end{enumerate}

\noindent
Note that in order to replace Rule
\eqref{eq:negative-1} by Rule \eqref{eq:negative:back}
we have to assume that the intruder's knowledge is a set of intruder facts
without repeated elements, i.e., the union operator \verb!_,_!
is $ACUI$ 
(associative-commutative-identity-idempotent).  This is completeness-preserving, since it is in line with the
restriction in \cite{EscMeaMes-tcs06} that the intruder learns a term only once. 

Furthermore, one might imagine that Rule \eqref{eq:positiveLearn-4}
and rules of type \eqref{eq:newstrand}
must also be modified in order to remove the \inI{M} expression from the 
intruder's knowledge of the right-hand side of each rule. However, this is not so,
since, by keeping the expresion \inI{M}, we force the backwards application of the rule
only when there is indeed a message for the intruder to be learned.
This provides some form of on-demand evaluation of the protocol.

The completeness proof for this optimization is trivial and thus omitted.
However, since we have modified the set of rules used for backwards reachability, we prove that such modification has the same reachability capabilities.
%
The set of rewrite rules actually
used for backwards narrowing 
is
$\overline{\caRP} = \{ \eqref{eq:negative:back},\eqref{eq:positiveNoLearn-2},\eqref{eq:positiveLearn-4} \}
\cup\eqref{eq:newstrand}$.
The following result ensures that
$\caRP$ and $\overline{\caRP}$ compute similar initial states
by backwards reachability analysis.
Its proof
is straightforward.

\begin{definition}[Inclusion]\label{def:inclusion}
Given a topmost rewrite theory $\caR_\caP = (\Symbols_\caP,E_{\caP},R_{\caP})$
representing protocol $\caP$,
and
two states $St_1,St_2$, 
we abuse notation and 
write $St_1 \subseteq_{E_\caP} St_2$
to denote that every state element (i.e., strand or intruder fact)
in $St_1$ 
appears in $St_2$ (modulo $E_\caP$).
\end{definition}

\begin{proposition}
Let $\caR_\caP = (\Symbols_\caP,E_{\caP},R_{\caP})$
be a topmost rewrite theory 
representing protocol $\caP$.
Let $St = ss\, \&\, SS\, \&\, (ik,IK)$
where $ss$ 
is a term representing a 
set of strands,
$ik$ 
is a term representing a set of intruder facts,
$SS$ is a variable for strands,
and 
$IK$ is a variable for intruder knowledge.
Let $St' = ss \,\&\, ik$.
If there is an initial state $St_{ini}$
and a substitution $\sigma$
such that
$St \narrto^*_{\sigma,R_{\caP}^{-1},E_\caP} St_{ini}$,
then
there is an initial state $St'_{ini}$
and two substitutions $\sigma'$, $\rho$
such that
$St' \narrto^*_{\sigma',\overline{R_{\caP}}^{-1},E_\caP} St'_{ini}$,
$\sigma \congr{E_\caP} \sigma'\circ\rho$,
and
$\rho(St'_{ini}) \subseteq_{E_\caP} St_{ini} $.
\end{proposition}

\subsection{Partial Order Reduction Giving Priority to Input Messages}\label{sec:input-first}

The different rewrite rules 
on which the backwards narrowing search from an attack pattern is based
are in general executed non-deterministically.
This is because the order of execution 
can
make a difference as to what subsequent rules can be executed.  
For example, an intruder cannot receive a term until 
it is sent by somebody, and that send action within a strand
may depend upon other receives in the past. 
There is one exception, 
Rule \eqref{eq:negative:back} (originally Rule \eqref{eq:negative-1}), 
which, in a backwards search, only
moves a negative term appearing right before the bar into the intruder's knowledge.  

\begin{example}
For instance, consider the attack pattern $(\dagger)$ in Example~\ref{exDH:cont}.
	\linebreak
Since the strand in the attack pattern has the input message
	\linebreak
$(e(exp(Y,n(B,r')),sec(a,r'')))^-$
but also has the intruder challenge
$\inI{sec(a,r'')}$,
there are several possible backwards narrowing steps: 
some processing the intruder challenge, and 
Rule \eqref{eq:negative:back} processing the input message.
\end{example}

The execution of 
Rule \eqref{eq:negative:back}
in a backwards search does not disable any other transitions; indeed, it only enables send transitions.  
Thus, it is safe to execute it at each stage \emph{before} any other transition.  
For the same reason,  if several applications of Rule~\ref{eq:negative:back} are possible, it is safe to execute them all at once before any other transition.
Requiring all executions of Rule~\ref{eq:negative:back} to execute first thus eliminates interleavings of
Rule~\ref{eq:negative:back} with send and receive transitions, which are equivalent to the case in which
Rule~\ref{eq:negative:back} executes first. 
In practice, this typically 
cuts down in half the search space size. 
The completeness proof for this optimization is trivial and is thus omitted.

Similar strategies have been employed by other tools in forward searches.  For example,  in
\cite{shmatikov-stern-csfw98}, a strategy is introduced that always executes 
send transitions  first whenever they are enabled.
Since a send transition does not depend on any other component 
of the state in order to take place, it can safely be executed first.  The original NPA also used this strategy; it had
a receive transition (similar to the input message in Maude-NPA)
which had the effect of adding new terms to the intruder's knowledge, and which always was executed before any other transition 
once it was enabled.

\subsection{Early Detection of Inconsistent States}\label{sec:inconsistent}

There are several types of states that are always unreachable or inconsistent.
%

\begin{example}\label{ex:ddagger}
Consider again the attack pattern $(\dagger)$ in Example~\ref{exDH:cont}. 
%
%
After a couple of backwards narrowing steps, the Maude-NPA finds
the following state, 
where the intruder learns
$e(exp(Y,n(B,r')),sec(a,r''))$
by assuming she can learn
$exp(Y,n(B,r'))$
and
$sec(a,r'')$ and combines them:

{\small
$$
\begin{array}{ll}
\begin{array}{@{}l@{}}
[\ nil \ |\ 
(exp(Y,n(B,r')))^-, 
(sec(a,r''))^-, 
(e(exp(Y,n(B,r')),sec(a,r'')))^+\ ]\ \&\\[1ex]
::r':: \\[1ex]
[\ (A ; B ; Y)^-, (B ; A ; exp(g,n(B,r')))^+ \ |\ 
(e(exp(Y,n(B,r')),sec(a,r'')))^-\ ]\ \&\\[1ex]
(\inI{sec(a,r'')},\ 
\inI{exp(Y,n(B,r'))},\  
\nI{e(exp(Y,n(B,r')),sec(a,r''))})
\end{array}
& (\ddagger)
\end{array}
$$
}

\noindent
From this state, the intruder tries to learn $sec(a,r'')$
by assuming she can learn
messages
$(e(exp(Y,n(B,r')),sec(a,r'')))$
and 
$exp(Y,n(B,r'))$ and combines them in a decryption:

{\small
$$
\begin{array}{l}
[\ nil \ |\ 
(exp(Y,n(B,r')))^-, 
(e(exp(Y,n(B,r')),sec(a,r'')))^-,
(sec(a,r''))^+\ ]\ \&\\[1ex]
[\ nil \ |\ 
(exp(Y,n(B,r')))^-, 
(sec(a,r''))^-, 
(e(exp(Y,n(B,r')),sec(a,r'')))^+\ ]\ \&\\[1ex]
::r':: \\[1ex]
[\ (A ; B ; Y)^-, (B ; A ; exp(g,n(B,r')))^+ \ |\ 
(e(exp(Y,n(B,r')),sec(a,r'')))^-\ ]\ \&\\[1ex]
(\inI{sec(a,r'')},\ 
\inI{exp(Y,n(B,r'))}, \\[1ex]
\inI{e(exp(Y,n(B,r')),sec(a,r''))},\ 
\nI{e(exp(Y,n(B,r')),sec(a,r''))})
\end{array}
$$
}

\noindent
But then this state is inconsistent,
since we have both
the challenge
	\linebreak
$\inI{e(exp(Y,n(B,r')),sec(a,r''))}$
and
the already learned message
	\linebreak
$\nI{e(exp(Y,n(B,r')),sec(a,r''))})$ at the same time,
violating the \emph{learn-only-once} condition in Maude-NPA.
\end{example}

If the Maude-NPA attempts to search beyond an inconsistent state, 
it will never find an initial state.  
For this reason, the Maude-NPA search strategy
always marks the following types of states as unreachable, and 
does not search beyond them any further:
\begin{enumerate}
\item\label{failure1} A state $St$  containing
two contradictory facts $\inI{t}$ and $\nI{t}$ (modulo $E_\caP$)
for a term $t$.
\item\label{failure2} A state $St$ whose intruder's knowledge contains
the fact $\nI{t}$ and 
a strand of the form
$[m_1^\pm,\linebreak[2] \ldots,\linebreak[2] t^-,\linebreak[2] \ldots,\linebreak[2] m_{j-1}^\pm \mid\linebreak[2] m_{j}^\pm,\linebreak[2] \ldots,\linebreak[2] m_k^\pm]$
(modulo $E_\caP$).
\item\label{failure3} A state $St$ containing
a fact $\inI{t}$ such that $t$ contains a fresh variable $r$
and the strand in $St$ indexed by $r$,
i.e., 
$::r_1,\ldots,r,\ldots,r_k::\;
[m_1^\pm,\linebreak[2] \ldots,\linebreak[2] m_{j-1}^\pm \mid\linebreak[2] m_{j}^\pm,\linebreak[2] \ldots,\linebreak[2] m_k^\pm]$,
cannot produce $r$,
i.e.,
$r$ is not a subterm of any output message in $m_1^\pm, \ldots, m_{j-1}^\pm$.
\item\label{failure4} A state $St$ 
containing 
a strand of the form
$[m_1^\pm,\linebreak[2] \ldots,\linebreak[2] t^-,\linebreak[2] \ldots,\linebreak[2] m_{j-1}^\pm \mid\linebreak[2] m_{j}^\pm,\linebreak[2] \ldots,\linebreak[2] m_k^\pm]$
for some term $t$ such that $t$ contains a fresh variable $r$
and the strand in $St$ indexed by $r$
cannot produce $r$.
\end{enumerate}
Note that case \ref{failure2} will become an instance of 
case \ref{failure1} after some backwards narrowing steps, 
and the same happens with cases
\ref{failure4} and \ref{failure3}.
The proof of inconsistency of cases \ref{failure1} and \ref{failure3}
is straightforward. 

\subsection{Transition Subsumption}\label{sec:folding}

Partial order reduction (POR) techniques  
are common in state exploration.
However, POR techniques for narrowing-based state exploration
do not seem to 
have been explored in detail, although they may be extremely relevant
and 
may 
afford greater reductions
than in standard state exploration
based on ground terms rather than on terms with logical variables.
For instance, the simple concept of two states being equivalent modulo
renaming of variables does not apply to standard state exploration,
whereas it does apply to narrowing-based state exploration.
In \cite{escobar-meseguer-rta07}, 
Escobar and Meseguer studied 
narrowing-based state exploration and POR techniques,
which may transform 
an infinite-state system into a finite one.
However, the Maude-NPA needs a dedicated POR technique
applicable to its
specific execution model. 

Let us motivate this POR technique with an example before giving a more detailed explanation.

\begin{example}
Consider again the attack pattern $(\dagger)$ in Example~\ref{exDH:cont}. 
%
%
After a couple of backwards narrowing steps, the Maude-NPA finds
the state $(\ddagger)$ of Example~\ref{ex:ddagger}: 

{\small
$$
\begin{array}{l}
 [\ nil \mid exp(Y,n(B,r'))^-, sec(a,r'')^-, (e(exp(Y,n(B,r')),sec(a,r'')))^+\ ]\ \& \\[1ex]
::r':: \\[1ex]
[\ (A ; B ; Y)^-, (B ; A ; exp(g,n(B,r')))^+ \ |\ 
(e(exp(Y,n(B,r')),sec(a,r'')))^-\ ]\ \& \\[1ex]
( \inI{sec(a,r'')},\ 
\inI{exp(Y,n(B,r'))},\ 
\nI{e(exp(Y,n(B,r')),sec(a,r''))})
\end{array}
$$
}

\noindent
However, the following state is also generated after a couple of narrowing steps from the attack pattern,
where,
thanks to the equational theory,
variable $Y$ is instantiated 
to $exp(G,N)$ for $G$ a generator --indeed the constant $g$---
and $N$ a nonce variable:

{\small
$$
\begin{array}{l}
 [\ nil \mid exp(G,n(B,r'))^-, N^-, exp(G,N * n(B,r'))^+\ ]\ \& \\[1ex]
 [\ nil \mid exp(G,N * n(B,r'))^-, sec(a,r'')^-, (e(exp(G,N * n(B,r')),sec(a,r'')))^+\ ]\ \& \\[1ex]
::r':: 
[\ (A ; B ; exp(G,N))^-, (B ; A ; exp(g,n(B,r')))^+ \\[1ex]
\hspace{5.5cm}
\ |\ 
(e(exp(G,N * n(B,r')),sec(a,r'')))^-\ ]\ \& \\[1ex]
( \inI{sec(a,r'')},\ 
\inI{exp(G,n(B,r')},\ 
\inI{N},\\[1ex] 
\nI{exp(G,N * n(B,r')},\ 
\nI{e(exp(G,N * n(B,r')),sec(a,r''))})
\end{array}
$$
}

\noindent
However, 
the unreachability of the second state 
is implied (modulo $E_\caP$) 
by the unreachability of the first state;
unreachability in the sense of Definition~\ref{def:unreachable}.
Intuitively,
the challenges present in the first state that are relevant for  backwards reachability are included in the second state, namely, the 
challenges \inI{sec(a,r'')} and \inI{exp(Y,n(B,r')}. 
Indeed, 
the unreachability of the following ``kernel'' state 
implies the unreachability of both states, 
although this kernel state is never computed by the Maude-NPA:

{\small
$$
\begin{array}{l}
::r':: 
[\ (A ; B ; Y)^-, (B ; A ; exp(g,n(B,r')))^+ \ |\ 
(e(exp(Y,n(B,r')),sec(a,r'')))^-\ ]\ \& \\[1ex]
(\inI{sec(a,r'')},\ \inI{exp(Y,n(B,r')})
\end{array}
$$
}

\noindent
Note that the converse is not true, i.e., the second state does not
imply the first one, since it contains one more intruder item
relevant for backwards reachability purposes, namely 
\inI{N}.
\end{example}

Let us now formalize this state space reduction 
and prove its completeness.
First, an auxiliary relation $St_1 \triangleright St_2$ identifying
whether $St_1$ is smaller than $St_2$ in terms of messages to be learned by the intruder.

\begin{definition}
Given a topmost rewrite theory $\caR_\caP = (\Symbols_\caP,E_{\caP},R_{\caP})$
representing protocol $\caP$,
and two non-initial states 
$St_{1}$
and
$St_{2}$,
we write 
$St_{1} \triangleright St_{2}$ 
(or $St_{2} \triangleleft St_{1}$)
if each intruder fact of the form $\inI{t}$
in $St_1$ appears in $St_2$ (modulo $E_\caP$)
and each non-initial strand
in $St_1$ appears in $St_2$ (modulo $E_\caP$ 
and 
with the vertical bar at the same position).
\end{definition}

Then, we define the relation $St_1 \blacktriangleright St_2$ 
which extends $St_1 \triangleright St_2$ to the case where $St_1$
is more general than $St_2$ w.r.t. variable instantiation.

\begin{definition}[$\caP$-subsumption relation]\label{def:subsumption}
Given a topmost rewrite theory 
	\linebreak
$\caR_\caP = (\Symbols_\caP,E_{\caP},R_{\caP})$
representing protocol $\caP$
and two non-initial states $St_{1}, St_{2}$,
we write $St_{1} \blacktriangleright St_{2}$
(or $St_{2} \blacktriangleleft St_{1}$)
and say that $St_{2}$ is \emph{$\caP$-subsumed} by $St_{1}$
if there is a substitution $\theta$ s.t. 
$\theta(St_{1}) \triangleright St_{2}$.
\end{definition}

\noindent
Note that we restrict the relation $\blacktriangleright$
to non-initial states because, otherwise, an initial state will imply any other state, erroneously making the search space finite after an initial state has been found. 

The following results provide the appropriate connection
between 
	\linebreak
$\caP$-subsumption and narrowing transitions.
%
%
%
First, we consider the simplest case where,
given two non-initial states $St_1,St_2$
such that
$St_{1} \blacktriangleright St_{2}$, 
a narrowing step on $St_2$, yielding state $St'_2$,
does not affect the transition subsumption property $\blacktriangleright$
and thus
$St_{1} \blacktriangleright St'_{2}$.
The proof is straightforward.

\begin{lemma}\label{lem:subsumption-1}
Given a topmost rewrite theory $\caR_\caP = (\Symbols_\caP,E_{\caP},R_{\caP})$
representing protocol $\caP$
and two non-initial states $St_{1}, St_{2}$. 
If 
(i) there is a substitution $\theta$ s.t.
$\theta(St_{1}) \triangleright St_{2}$,
i.e., $St_{1} \blacktriangleright St_{2}$,
(ii) there is a narrowing step 
$St_{2} \narrto_{\sigma_{2},R_{\caP}^{-1},E_\caP} St'_{2}$,
(iii)
each intruder fact of the form $\inI{t}$
in $\sigma_2(\theta(St_1))$ 
appears in $St'_2$ (modulo $E_\caP$)
and 
(iv)
each non-initial strand
in $\sigma_2(\theta(St_1))$ 
appears in $St'_2$ (modulo $E_\caP$),
then 
$\sigma_2(\theta(St_{1})) \triangleright St'_{2}$,
i.e.,
$St_{1} \blacktriangleright St'_{2}$.
\end{lemma}

Second, we consider what happens when,
given two non-initial states $St_1,St_2$
such that
$St_{1} \blacktriangleright St_{2}$, 
a narrowing step on $St_2$, yielding state $St'_2$,
does affect the transition subsumption property $\blacktriangleright$
and thus
$St_{1} \not\blacktriangleright St'_{2}$.
The proof is straightforward.

\begin{lemma}\label{lem:subsumption-2}
Given a topmost rewrite theory $\caR_\caP = (\Symbols_\caP,E_{\caP},R_{\caP})$
representing protocol $\caP$
and two non-initial states $St_{1}, St_{2}$. 
If 
(i) there is a substitution $\theta$ s.t.
$\theta(St_{1}) \triangleright St_{2}$,
i.e., $St_{1} \blacktriangleright St'_{2}$,
(ii) there is a narrowing step 
$St_{2} \narrto_{\sigma_{2},R_{\caP}^{-1},E_\caP} St'_{2}$,
and 
(iii)
$\sigma_2(\theta(St_{1})) \,{\not\triangleright}\, St'_{2}$,
i.e., $St_{1} \not\blacktriangleright St'_{2}$,
then 
either
(a)
there is an intruder fact of the form $\inI{t}$
in $\sigma_2(\theta(St_1))$ 
that 
does not appear in $St'_2$ (modulo $E_\caP$),
or
(b)
there is a non-initial strand
in $\sigma_2(\theta(St_1))$ 
that
does not appear in $St'_2$ (modulo $E_\caP$).
\end{lemma}

Now, we can consider both cases 
of Lemma~\ref{lem:subsumption-2} separately:
either an expression $\inI{t}$ in $St'_2$ 
or a non-initial strand in $St'_2$,
not appearing in the instantiated version of $St_1$.
First, 
the case where 
an expression $\inI{t}$ in $St'_2$ 
does not appear in the instantiated version of $St_1$.

\begin{lemma}\label{lem:subsumption-2-inI}
Given a topmost rewrite theory $\caR_\caP = (\Symbols_\caP,E_{\caP},R_{\caP})$
representing protocol $\caP$
and two non-initial states $St_{1}, St_{2}$. 
If 
(i) there is a substitution $\theta$ s.t.
$\theta(St_{1}) \triangleright St_{2}$,
i.e., $St_{1} \blacktriangleright St'_{2}$,
(ii) there is a narrowing step 
$St_{2} \narrto_{\sigma_{2},R_{\caP}^{-1},E_\caP} St'_{2}$,
and 
(iii)
there is an intruder fact of the form $\inI{t}$
in $\sigma_2(\theta(St_1))$ 
that 
does not appear in $St'_2$ (modulo $E_\caP$),
then
(a)
$\nI{t}$
does appear in $St'_2$ (modulo $E_\caP$)
and
(b)
there is a state $St'_{1}$ and a substitution $\sigma_{1}$
such that
$St_{1} \narrto_{\sigma_{1},R_{\caP}^{-1},E_\caP} St'_{1}$
and 
either 
$St'_1$ is an initial state
or
there is a substitution $\rho$ s.t.
$\rho(St'_{1}) \triangleright St'_{2}$,
i.e., $St'_{1} \blacktriangleright St'_{2}$,
\end{lemma}
\begin{proof}
We prove the result by considering the different rules applicable to $St_{2}$ (remember that in $\caR$, rewriting and narrowing steps always happen 
at the top position).
Note that property (a) is immediate because rules in $R_\caP$ do not remove expressions of the form $\inI{m}$.
Note also that if 
$\inI{t}$
does appear in $St_2$ (modulo $E_\caP$)
and
$\nI{t}$
does appear in $St'_2$ (modulo $E_\caP$),
then only Rule
\eqref{eq:positiveLearn-4}
or rules of type
\eqref{eq:newstrand}
have been applied to $St_2$ as follows:

\begin{itemize}

\item Reversed version of Rule \eqref{eq:positiveLearn-4}, i.e.,
$St_{2} \narrto_{\sigma_{2},R_{\caP}^{-1},E_\caP} St'_{2}$ using the following rule
$$
[L, M^+ ~|~ L']{\,\&\,}\linebreak[2]SS{\,\&\,}\linebreak[2](\inI{M},IK) 
\to 
[L ~|~ M^+, L']{\,\&\,}\linebreak[2]SS{\,\&\,}\linebreak[2](\nI{M},IK)
.$$ 
Recall that 
there is an intruder fact 
in $\sigma_2(\theta(St_1))$ 
of the form $\inI{t}$
for $t$ a message term
that 
does not appear in $St'_2$ (modulo $E_\caP$)
and 
$t \congr{E_\caP} \sigma_2(M)$. 
Thus,
$\inI{\sigma_2(M)}$
does appear in $\sigma_2(\theta(St_1))$ (modulo $E_\caP$).
Here we have several cases:

\begin{itemize}

\item
If the strand $\sigma_2([L, M^+ ~|~ L'])$
appears in $\sigma_2(\theta(St_{1}))$,
then the very same narrowing step can be performed on $St_{1}$,
i.e.,
there exist $\sigma_1,\rho$ such that
$St_{1} \narrto_{\sigma_{1},R_{\caP}^{-1},E_\caP} St'_{1}$
with the same rule and 
$\theta \circ \sigma_2 \congr{E_\caP} \sigma_1\circ\rho$.
Thus, 
either $St'_{1}$ is an initial state
or
$\rho(St'_{1}) \triangleright St'_{2}$,
since:
(i)
each positive intruder fact 
in $\sigma_2(\theta(St_1))$ of the form $\inI{u}$ for $u$ a message term,
except $\inI{\sigma_2(M)}$,
appears in $\rho(St'_1)$ (modulo $E_\caP$),
(ii)
$\nI{\sigma_2(M)}$
appears in $\rho(St'_1)$ (modulo $E_\caP$),
(iii)
each non-initial strand
in $\sigma_2(\theta(St_1))$, 
except $\sigma_2([L, M^+ {~|~} L'])$,
has not been modified and
appears in $\rho(St'_1)$ as well (modulo $E_\caP$),
and 
(iv) for 
$\sigma_2([L, M^+ {~|~} L'])$ in $\sigma_2(\theta(St_1))$,
$\rho'([L \mid M^+, L'])$ appears in $\rho(St'_1)$
and in $St'_2$.

\item
If the strand $\sigma_2([Lm M^+ ~|~ L'])$
does not appear in $\sigma_2(\theta(St_{1}))$,
then the strand $\sigma_2([L, M^+ ~|~ L'])$ 
corresponds to a strand $\caS_\caP$ in the protocol specification
that had been introduced 
via a rule of the set \eqref{eq:newstrand}, where 
the strand's bar was clearly more to the right than in $\sigma_2([L, M^+ ~|~ L'])$. 
Note that it cannot correspond to a strand included originally in the attack pattern, because we assume that $St_{1}$ and $St_{2}$ are states generated 
by backwards narrowing from the same attack state and then both $St_1$ and $St_2$ should have the strand.
Therefore, 
since the strand $\sigma_2([L, M^+ ~|~ L'])$
corresponds to a strand in $\caS_\caP$
and 
the set \eqref{eq:newstrand}
contains a rewrite rule for each strand 
of the form $[~ l_1,\ u^+,\ l_2~]$
in $\caS_\caP$,
there must be a rule $\alpha$ in \eqref{eq:newstrand}
introducing a strand of the form $[~ l_1,\ u^+,\ l_2~]$ 
and
there must be substitutions $\sigma_1,\rho$ such that
$St_{1} \narrto_{\sigma_{1},R_{\caP}^{-1},E_\caP} St'_{1}$
using the rule $\alpha$ 
and 
$\theta \circ \sigma_2 \congr{E_\caP} \sigma_1\circ\rho$.
Thus, 
either $St'_{1}$ is an initial state
or
$\rho(St'_{1}) \triangleright St'_{2}$,
since:
(i)
each positive intruder fact 
in $\sigma_2(\theta(St_1))$
of the form $\inI{u}$
for $u$ a message term,
except $\inI{\sigma_2(M)}$,
appears in $\rho(St'_1)$ (modulo $E_\caP$),
(ii)
$\nI{\sigma_2(M)}$
appears in $\rho(St'_1)$ (modulo $E_\caP$),
(iii)
each non-initial strand
in $\sigma_2(\theta(St_1))$
has not been modified and
appears in $\rho(St'_1)$ as well (modulo $E_\caP$),
and 
(iv) 
$\sigma_2([~ l_1 \mid u^+,\ l_2~])$ appears in $\rho(St'_1)$
and in $St'_2$.

\end{itemize}

\item Rules in \eqref{eq:newstrand}, i.e.,
$St_{2} \narrto_{\sigma_{2},R_{\caP}^{-1},E_\caP} St'_{2}$ using a
rule of the form
$$
\{
SS{\,\&\,}\linebreak[2](\inI{u},IK) \to 
[l_1 \mid u^+, l_2]{\,\&\,}\linebreak[2]SS{\,\&\,}\linebreak[2](\nI{u},IK)
  \mid\linebreak[2] [l_1, u^+, l_2] \in \caP\}
.$$
Recall that 
there is an intruder fact 
in $\sigma_2(\theta(St_1))$ 
of the form $\inI{t}$
for $t$ a message term
that 
does not appear in $St'_2$ (modulo $E_\caP$)
and 
$t \congr{E_\caP} \sigma_2(u)$, where $u$ is the message term
used by the rewrite rule.
Thus, 
$\inI{\sigma_2(u)}$
does appear in $\sigma_2(\theta(St_1))$ (modulo $E_\caP$).
That is, the same narrowing step is available from  $\sigma_2(\theta(St_1))$
and 
there exist $\sigma_1,\rho$ such that
$St_{1} \narrto_{\sigma_{1},R_{\caP}^{-1},E_\caP} St'_{1}$
with the same rule and 
$\theta \circ \sigma_2 \congr{E_\caP} \sigma_1\circ\rho$.
Thus, 
either $St'_{1}$ is an initial state
or
$\rho(St'_{1}) \triangleright St'_{2}$.
\end{itemize}
This concludes the proof.
\qed
\end{proof}

Second, 
the case where 
a non-initial strand in $St'_2$
does not appear in the instantiated version of $St_1$.

\begin{lemma}\label{lem:subsumption-2-strand}
Given a topmost rewrite theory $\caR_\caP = (\Symbols_\caP,E_{\caP},R_{\caP})$
representing protocol $\caP$
and two non-initial states $St_{1}, St_{2}$. 
If 
(i) there is a substitution $\theta$ s.t.
$\theta(St_{1}) \triangleright St_{2}$,
(ii) there is a narrowing step 
$St_{2} \narrto_{\sigma_{2},R_{\caP}^{-1},E_\caP} St'_{2}$,
and 
(iii)
there is a non-initial strand 
$[m_1^\pm,\ldots,m_i^\pm\mid m_{i+1}^\pm,\ldots,m_n^\pm]$
in $\sigma_2(\theta(St_1))$ 
that
does not appear in $St'_2$ (modulo $E_\caP$),
then
(a) $\sigma_2\restrict{\var(St_2)} = \idsubst$,
(b)
$[m_1^\pm,\ldots,m_{i-1}^\pm\mid m_{i}^\pm,\ldots,m_n^\pm]$
does appear in $St'_2$ (modulo $E_\caP$)
and
(c)
there is a state $St'_{1}$
such that
$St_{1} \narrto_{\idsubst,R_{\caP}^{-1},E_\caP} St'_{1}$
and 
either 
$St'_1$ is an initial state
or
$St'_{1} \triangleright St'_{2}$.
\end{lemma}
\begin{proof}
We prove the result by considering the different rules applicable to $St_{2}$ (remember that in $\caR$, rewriting and narrowing steps always happen 
at the top position).
Note that property (a) is immediate because rules in $R_\caP$ do not remove strands, only move the vertical bar to the left of the sequences of messages in the strands.
Note also that if 
$[m_1^\pm,\ldots,m_i^\pm\mid m_{i+1}^\pm,\ldots,m_n^\pm]$
appears in $\sigma_2(\theta(St_1))$ 
and
$[m_1^\pm,\ldots,m_{i-1}^\pm\mid m_{i}^\pm,\ldots,m_n^\pm]$
appears in $St'_2$,
then only Rule
\eqref{eq:positiveNoLearn-2}
or Rule
\eqref{eq:negative:back}
have been applied to $St_2$ as follows:

\begin{itemize}

\item\label{proof:eq:positiveNoLearn} 
Reversed version of Rule \eqref{eq:positiveNoLearn-2}, i.e.,
$St_{2} \narrto_{\sigma_{2},R_{\caP}^{-1},E_\caP} St'_{2}$ using the following rule
$$
[L, M^+ {~|~} L']{\,\&\,}\linebreak[2]SS{\,\&\,}\linebreak[2]IK 
\to 
[L {~|~} M^+, L']{\,\&\,}\linebreak[2]SS{\,\&\,}\linebreak[2]IK
.$$

\item Reversed version of Rule \eqref{eq:negative:back}, i.e.,
$St_{2} \narrto_{\sigma_{2},R_{\caP}^{-1},E_\caP} St'_{2}$ using the following rule
$$
[L, M^- ~|~ L']{\,\&\,}\linebreak[2]SS{\,\&\,}\linebreak[2]IK
\to
[L ~|~ M^-, L']{\,\&\,}\linebreak[2]SS{\,\&\,}\linebreak[2](\inI{M},IK) 
. $$

\end{itemize}
However, 
note that $\sigma_2\restrict{\var(St_2)} = \idsubst$ in both possible rewrite steps.
Then, 
there is a state $St'_1$ such that
$St_{1} \narrto_{\idsubst,R_{\caP}^{-1},E_\caP} St'_{1}$
with the same rule 
and it is straightforward that
either 
$St'_1$ is an initial state
or
$St'_{1} \blacktriangleright St'_{2}$, 
since 
only 
the vertical bar has been moved.
\qed
\end{proof}

Now we can formally define
the relation
between 
$\caP$-subsumption and one narrowing step.
In the following,
$\narrto^{\{0,1\}}_{\sigma,R_{\caP}^{-1},E_\caP}$ 
denotes zero or one narrowing steps.

\begin{lemma}\label{lem:subsumption}
Given a topmost rewrite theory $\caR_\caP = (\Symbols_\caP,E_{\caP},R_{\caP})$
representing protocol $\caP$
and two non-initial states $St_{1}, St_{2}$. 
If $St_{1} \blacktriangleright St_{2}$ 
and
	\linebreak
$St_{2} \narrto_{\sigma_{2},R_{\caP}^{-1},E_\caP} St'_{2}$,
then there is a state $St'_{1}$ and a substitution $\sigma_{1}$
such that
$St_{1} \narrto^{\{0,1\}}_{\sigma_{1},R_{\caP}^{-1},E_\caP} St'_{1}$
and 
either 
$St'_1$ is an initial state
or
$St'_{1} \blacktriangleright St'_{2}$.
\end{lemma}
\begin{proof}
Since $St_{1} \blacktriangleright St_{2}$,
there is a substitution $\theta$ s.t.
$\theta(St_{1}) \triangleright St_{2}$.
If 
each intruder fact of the form $\inI{t}$
in $\sigma_2(\theta(St_1))$ 
appears in $St'_2$ (modulo $E_\caP$)
and 
each non-initial strand
in $\sigma_2(\theta(St_1))$ 
appears in $St'_2$ (modulo $E_\caP$),
then, by Lemma~\ref{lem:subsumption-1},
$\sigma_2(\theta(St_{1})) \triangleright St'_{2}$,
i.e., $St_{1} \blacktriangleright St'_{2}$.
Otherwise,
Lemma~\ref{lem:subsumption-2} states that 
either
(a)
there is an intruder fact of the form $\inI{t}$
in $\sigma_2(\theta(St_1))$ 
that 
does not appear in $St'_2$ (modulo $E_\caP$),
or
(b)
there is a non-initial strand
in $\sigma_2(\theta(St_1))$ 
that
does not appear in $St'_2$ (modulo $E_\caP$).
For case (a), 
by Lemma~\ref{lem:subsumption-2-inI}, 
there is a state $St'_{1}$ and a substitution $\sigma_{1}$
such that
$St_{1} \narrto_{\sigma_{1},R_{\caP}^{-1},E_\caP} St'_{1}$
and 
either 
$St'_1$ is an initial state
or
there is a substitution $\rho$ s.t.
$\rho(St'_{1}) \triangleright St'_{2}$.
For case (b), 
by Lemma~\ref{lem:subsumption-2-strand}, 
$\sigma_2\restrict{\var(St_2)} = \idsubst$,
and
there is a state $St'_{1}$
such that
$St_{1} \narrto_{\idsubst,R_{\caP}^{-1},E_\caP} St'_{1}$
and 
either 
$St'_1$ is an initial state
or
$St'_{1} \triangleright St'_{2}$,
i.e., $St'_{1} \blacktriangleright St'_{2}$.
\qed
\end{proof}

Preservation of reachability follows from the following main theorem.
Note 
that the relation $\blacktriangleright$ is applicable only
to non-initial states, whereas 
the relation $\subseteq_{E_\caP}$ of Definition~\ref{def:inclusion} is applicable
to both initial and non-initial states.

\begin{theorem}\label{thm:subsumption}
Given a topmost rewrite theory $\caR_\caP = (\Symbols_\caP,E_{\caP},R_{\caP})$
representing protocol $\caP$
and two states $St_{1}, St_{2}$. 
If $St_{1} \blacktriangleright St_{2}$,
$St^{ini}_2$ is an initial state, 
and
$St_{2} \narrto^*_{\sigma_{2},R_{\caP}^{-1},E_\caP} St^{ini}_{2}$,
then there is an initial state $St^{ini}_{1}$ and substitutions $\sigma_{1}$
and $\theta$
such that
$St_{1} \narrto^*_{\sigma_{1},R_{\caP}^{-1},E_\caP} St^{ini}_{1}$,
and $\theta(St^{ini}_{1}) \subseteq_{E_\caP} St^{ini}_{2}$.
\end{theorem}
\begin{proof}
Consider
$St_{2} = U_0$,
$St^{ini}_{2} = U_n$,
$\sigma_2 = \rho_1\cdots\rho_n$,
and 
$U_0 \narrto^n_{\rho_{i},R_{\caP}^{-1},E_\caP} U_n$.
Note that $n\neq 0$, since 
$St_2$ cannot be an initial state because
$St_{1} \blacktriangleright St_{2}$
implies 
that
both $St_1$ and $St_2$ are not initial states.
Then, 
by Lemma~\ref{lem:subsumption},
there is $j \leq n$ such that
for each $i < j$,
$U_{i-1} \narrto_{\rho_{i},R_{\caP}^{-1},E_\caP} U_{i}$
and there is a step
$U'_{i-1} \narrto_{\rho'_{i},R_{\caP}^{-1},E_\caP} U'_{i}$
s.t.
$U'_i \blacktriangleright U_i$.
Note that $U'_j$ is an initial state
and
there is a substitution $\theta$ s.t.
$\theta(U'_j) \subseteq_{E_\caP} U_j \subseteq_{E_\caP} U_n$.
\qed
\end{proof}

This POR technique is used as follows:
we keep all the states of the backwards narrowing-based tree and compare each new leaf node
of the tree with all the previous states in the tree.
If a leaf node is $\caP$-subsumed by a previously generated node in the tree, 
we discard such leaf node.


\subsection{The Super-Lazy Intruder}\label{sec:super-lazy-intruder}

Sometimes terms appear in the intruder's knowledge that are trivially learnable by the intruder.  These
include terms initially available to the intruder (such as names) and variables.  
In the case of variables, specially,
the intruder can substitute any arbitrary term of the same sort as the variable,%
\footnote{This, of course, is subject to the assumption that the intruder can produce at least one term of that sort.  But since the intruder is assumed to have access to the network and to all the operations available to an honest principal, this is a reasonable restriction to make.} 
and so there is no need to try to determine all the ways in which the intruder can do this.  For this reason it is safe, at least temporarily, to drop these terms from the state.  We will refer to those terms as {\em (super) lazy intruder} terms.

\begin{example}\label{ex:super-lazy}
Consider again the attack pattern $(\dagger)$ in Example~\ref{exDH:cont}. 
%
%
After a couple of backwards narrowing steps, the Maude-NPA finds
the following state
that considers how the intruder can learn
$sec(a,r'')$ by assuming he can learn 
a message $e(K, sec(a, r''))$ and the key $K$:

{\small
$$
\begin{array}{ll}
\begin{array}{@{}l@{}}
[\ nil \mid K^-,\ e(K, sec(a, r'')))^-,\ sec(a, r'')^+\ ]\& \\[1ex]
::r':: \\[1ex]
[\ (A ; B ; Y)^-, (B ; A ; exp(g,n(B,r')))^+ \mid 
(e(exp(Y,n(B,r')),sec(a,r'')))^-\ ]\&\\[1ex]
(\inI{e(exp(Y,n(B,r')),sec(a,r''))},\ 
\inI{K},\ 
\inI{e(K, sec(a, r'')))},\ 
\nI{sec(a,r'')})
\end{array}
& (\natural)
\end{array}
$$
}

\noindent
Here variable $K$ is a super-lazy term and the tool wouldn't search for values.
The problem, of course, is that later on in the search the variable $K$
may become instantiated, in which case
the term then becomes relevant to the search.  
Indeed, after some more backwards narrowing steps,
the tool tries to unify message 
$e(K, sec(a, r'')))$
with an output message 
$e(exp(\overline{X},n(\overline{A},\overline{r})), sec(\overline{A},\overline{r_2}))$
of an explicitly added Bob's strand of the form

$$
\begin{array}{@{}l@{}}
:: \overline{r_1},\overline{r_2} ::\\[1ex]
[\ (\overline{A} ; \overline{B} ; exp(g,n(\overline{A},\overline{r_1})))^+,\ (\overline{B} ; \overline{A} ; \overline{X})^-,\  (e(exp(\overline{X},n(\overline{A},\overline{r})), sec(\overline{A},\overline{r_2})))^+]
\end{array}
$$

\noindent
thus getting an instantiation for the super-lazy term $K$,
namely 
	\linebreak
$\{K \mapsto exp(\overline{X},n(\overline{A},\overline{r}))\}$.
\end{example}

\noindent
Note that the tool might continue searching for an initial state
when a super lazy term is properly instantiated, and this would not cause the tool to prove an insecure protocol to be secure. However, it would lead to an unacceptably large number of false attacks
because the contents of variable $K$ are expected to be learned by the intruder too.

We take an approach similar to that of the lazy intruder of Basin et al. \cite{ofmc}
and extend it to a more general case, that we call 
\emph{super-lazy terms}.
We note that this use of what we here call
the super-lazy intruder was also present in the original NPA. 
%
%

The set $\caL(St)$ of super-lazy terms w.r.t. a state $St$ is inductively generated as a subset $\caL(St) \subseteq \caT_\Omega(Y \cup IK_0)$
where $IK_0$ is the basic set of terms known by the intruder at the beginning of a protocol execution,
$Y$ is a subset of the variables of $St$,
and $\Omega$ is the set of operations available to the intruder.
The idea of super-lazy terms is that we also want to exclude from
$\caL(St)$ the set $IK^{\not\in}(St)$ of terms that the intruder does not know and all its possible combinations with symbols in $\Omega$. 

\begin{definition}[Super-lazy terms]\label{def:superlazy}
Let 
$\caR_\caP = (\Symbols_\caP,E_{\caP},R_{\caP})$
be a topmost 
	\linebreak
rewrite theory 
representing protocol $\caP$.
Let $IK_0$ be the basic set of terms known by the intruder at the beginning of a protocol execution, defined as
	\linebreak
$IK_0=\{t' \mid [t^+] \in \caS_\caP,\ t' \congr{E_\caP} t\}$.
Let $\Omega$ be the set of operations available to the intruder,
defined as
$$\Omega = \{f: \sort{s_1} \cdots \sort{s_n} \to \sort{s} \mid [(X_1{:}\sort{s_1})^-,\ldots,(X_k{:}\sort{s_k})^-,(f(X_1{:}\sort{s_1},\ldots,X_k{:}\sort{s_k}))^+] \in \caS_\caP\}.$$
Let $St$ be a state (with logical variables).
Let $IK^{\not\in}(St)$ be
the set of terms that the intruder does not known at state $St$,
defined as $IK^{\not\in}(St) = \{m' \mid (\nI{m}) \in St,\ m' \congr{E_\caP} m\}$.
The
set $\caL(St)$ of \emph{super-lazy terms w.r.t. $St$}
(or simply super-lazy terms)
is defined as 
\begin{enumerate}
\item
$IK_0 \subseteq \caL(St)$,
\item
$\var(St) - IK^{\not\in}(St) \subseteq \caL(St)$,
\item
for each 
$f: \sort{s_1} \cdots \sort{s_n} \to \sort{s} \in \Omega$
and 
for all
$t_1{:}\sort{s_1},\ldots,t_k{:}\sort{s_k} \in \caL(St)$,
if 
$f(t_1{:}\sort{s_1},\ldots,t_k{:}\sort{s_k}) \not\in IK^{\not\in}(St)$,
then
$f(t_1{:}\sort{s_1},\ldots,t_k{:}\sort{s_k}) \in \caL(St)$.
\end{enumerate}
\end{definition}
\noindent 
The idea behind the super-lazy intruder is that, 
given a term made out of lazy intruder terms,
such as ``$a;e(K,Y)$'', where $a$ is a public name and $K$ and $Y$ are variables, 
the term ``$a;e(K,Y)$'' is also a (super) lazy intruder term
by applying the operations $e$ and $\_{;}\_$.

Let us first briefly explain how the (super) lazy intruder mechanism works before formally describing it.
A \emph{ghost state} is a state extended to allow expressions
of the form $\textnormal{ghost}(m)$ in the intruder's knowledge,
where $m$ is a super-lazy term.
When, during the backwards reachability analysis,
 we detect a state $St$ having a super lazy term $t$ in an expression
 $\inI{t}$ in the intruder's knowledge,
we replace the intruder fact $\inI{t}$ in $St$
by $ghost(t)$ and keep the ghost version of $St$ 
in the history
of states used by the transition subsumption of Section~\ref{sec:folding}.
For instance, the state 
$(\natural)$ of Example~\ref{ex:super-lazy} with a super-lazy intruder term $K$
would be represented as follows, where we have just replaced $\inI{K}$
by $ghost(K)$:

{\small
$$
\begin{array}{l}
[\ nil \mid K^-,\ e(K, sec(a, r'')))^-,\ sec(a, r'')^+\ ]\& \\[1ex]
::r':: 
[\ (A ; B ; Y)^-, (B ; A ; exp(g,n(B,r')))^+ \mid 
(e(exp(Y,n(B,r')),sec(a,r'')))^-\ ]\&\\[1ex]
(\textnormal{ghost}(K),\ 
\inI{e(exp(Y,n(B,r')),sec(a,r''))},\ 
\inI{e(K, sec(a, r'')))},\ 
\nI{sec(a,r'')})
\end{array}
$$
}

If later in the search tree we detect a ghost state $St'$ 
containing an expression $ghost(t)$ such that $t$ is no longer 
a super lazy intruder term, 
then 
there is a state $St$ 
with an expression $ghost(u)$
that precedes $St'$ in the narrowing tree
such that 
the message $u$
has been instantiated to $t$ in an appropriate way and 
we must reactivate such original state $St$.
That is, we ``roll back'' and replace the current state $St'$,
containing expression $ghost(t)$,
by 
an instantiated version of state $St$, namely $\theta(St)$,
where $t \congr{E_\caP} \theta(u)$.
This is explained in detail in Definition~\ref{def:resuscitation} below.

However, if the substitution $\theta$ binding variables in $u$ includes
variables of sort \sort{Fresh}, 
we have to keep them
in the reactivated version of $St$, since they are unique in our model.  Therefore, the strands indexed by these fresh variables
must also be included in the ``rolled back'' state, even if they were not there originally.
Moreover, they must have the bar at the place where it was when the strands were originally introduced.
We show below how this is accomplished.
Furthermore,  if any of the strands thus introduced have other variables of sort \sort{Fresh}
as subterms, then the strands indexed by those variables must be included too, and so on.
That is,
when a state $St'$ properly
instantiating a ghost expression $ghost(t)$ is found,
the procedure of rolling back
to the original state $St$
that gave rise to that ghost expression
implies not only applying the bindings for the variables of $t$ to $St$,
but also introducing in $St$ all the strands from $St'$ that produced
fresh variables and that either appear in the variables of $t$ or are recursively connected with them.

\begin{example}\label{ex:resuscitated}
For instance, after the tool finds an instantiation
for variable $K$, the tool rolls back to the state 
originating the super-lazy term $K$ as follows,
where 
we have copied the explicitly added Bob's strand with the vertical bar 
at the rightmost position because it is the strand generating
the \sort{Fresh} variable $r''$:

{\small
$$
\begin{array}{l}
[\ nil \mid exp(X,n(a,r))^-,\ e(exp(X,n(a,r), sec(a, r'')))^-,\ sec(a, r'')^+\ ]\& \\[1ex]
:: r,r'' ::\\[1ex]
[\ (a ; B' ; exp(g,n(a,r)))^+,\ (B' ; a ; X)^-,\  (e(exp(X,n(a,r)), sec(a,r'')))^+
\mid nil \ ]\&\\[1ex]
::r':: 
[\ (A ; B ; Y)^-, (B ; A ; exp(g,n(B,r')))^+ \mid 
(e(exp(Y,n(B,r')),sec(a,r'')))^-\ ]\&\\[1ex]
(\inI{e(exp(Y,n(B,r')),sec(a,r''))},\ 
\inI{exp(X,n(a,r))},\\[1ex]
\inI{e(exp(X,n(a,r)), sec(a, r'')))},\ 
\nI{sec(a,r'')})
\end{array}
$$
}
\end{example}

In order for the super-lazy intruder mechanism to be able to tell where the bar was when a strand was introduced,
we must modify the set of rules of type \eqref{eq:newstrand}
introducing new strands:
\begin{align}
&
\begin{array}[t]{@{}l@{}}
\{\ [\, l_1 \,|\, u^+] \,\&\, \{\nI{u},K\} \to \{\inI{u},K\}
\ \mid\   [\, l_1,\, u^+,\, l_2\,] \in \caS_\caP\}
\end{array}
\label{eq:newstrand:lazy}%
\end{align}

\noindent
Note that 
rules of type \eqref{eq:newstrand} introduce strands $[\,l_1 \mid u^+, l_2\,]$,
whereas here rules of type \eqref{eq:newstrand:lazy} 
introduce strands $[\,l_1 \mid u^+\,]$.
This slight modification makes it possible to safely move the position of the bar 
back to the place where the strand was introduced.
However, now the strands added may be \emph{partial}, since the whole
sequence of actions performed by the principal is not directly recorded
in the strand.
Therefore, the set of rewrite rules used by narrowing in reverse are now
$\widetilde{\caRP} = \{ \eqref{eq:negative:back},\eqref{eq:positiveNoLearn-2},\eqref{eq:positiveLearn-4} \}
\cup\eqref{eq:newstrand:lazy}$.

First, we define a new relation $\sqsubseteq_{E_\caP}$ between states,
which is similar to $\subseteq_{E_\caP}$ of Definition~\ref{def:inclusion}
but considers partial strands.
\begin{definition}[Partial Inclusion]
Given two states $St_1,St_2$, 
we abuse notation and 
write $St_1 \sqsubseteq_{E_\caP} St_2$
to denote that every intruder fact 
in $St_1$ 
appears in $St_2$ (modulo $E_\caP$)
and that every strand $[m_1^\pm, \ldots, m_k^\pm]$ in $St_1$,
either
appears in $St_2$ (modulo $E_\caP$)
or 
there is $i \in \{1,\ldots,k\}$ s.t.
$m_i^\pm = m_i^+$
and
$[m_1^\pm, \ldots, m_i^+]$ 
appears in $St_2$ (modulo $E_\caP$).
\end{definition}

\noindent
The following result ensures that if a state is reachable via backwards reachability analysis using $R_\caP$, then it is also reachable using $\widetilde{R_\caP}$.
Its proof
is straightforward.
\begin{proposition}
Let $\caR_\caP = (\Symbols_\caP,E_{\caP},R_{\caP})$
be a topmost rewrite theory 
representing protocol $\caP$.
Let $St = ss\, \&\, SS\, \&\, (ik,IK)$
where 
$ss$ is a term representing a set of strands,
$ik$ is a term representing a set of intruder facts,
$SS$ is a variable for strands,
and 
$IK$ is a variable for intruder knowledge.
	\linebreak
If there is an initial state $St_{ini}$
and a substitution $\sigma$
such that
	\linebreak
$St \narrto^*_{\sigma,R_{\caP}^{-1},E_\caP} St_{ini}$,
then
there is an initial state $St'_{ini}$
and two substitutions $\sigma'$, $\rho$
such that
$St \narrto^*_{\sigma',\widetilde{R_{\caP}}^{-1},E_\caP} St'_{ini}$,
$\sigma \congr{E_\caP} \sigma'\circ\rho$,
and
$\rho(St'_{ini}) \sqsubseteq_{E_\caP} St_{ini}$.
\end{proposition}

Now, we describe how to reactivate a state.
First, we formally define a ghost state.
\begin{definition}[Ghost State]
Given a topmost rewrite theory $\caR_\caP = \linebreak[4](\Symbols_\caP,E_{\caP},R_{\caP})$
representing protocol $\caP$
and a state $St$ containing an intruder fact $\inI{t}$ such that $t$ 
is a super-lazy term,
we define the \emph{ghost} version of $St$,
written 
$\overline{St}$,
by replacing $\inI{t}$ in $St$ by $ghost(t)$ in $\overline{St}$.
\end{definition}

Now, in order to resuscitate a state, we need to formally compute 
the strands that are generating \sort{Fresh} variables relevant
to the instantiation found for the super-lazy term.
\begin{definition}[Strand Reset]
Given a strand $s$ of the form
$::r_1,\ldots,r_k::\;[m_1^\pm,\ldots \mid \ldots,m_n^\pm]$,
when we want to move the bar to the rightmost position (denoting a final strand),
we write 
$s{\gg}=::r_1,\ldots,r_k::\;[m_1^\pm,\ldots,m_n^\pm \mid nil]$.
\end{definition}

\begin{definition}[Fresh Generating Strands]\label{def:strands:super-lazy-term}
Given a state $St$ containing an intruder fact 
$ghost(t)$ for some term $t$ with variables,
we define the set of strands associated to $t$, denoted 
$\textnormal{strands}_{St}(t)$, as follows:
\begin{itemize}
\item
for each strand $s$ in $St$ of the form
$::r_1,\ldots,r_k::\;[m_1^\pm,\ldots \mid \ldots,m_n^\pm]$,
if there is $i\in\{1,\ldots,k\}$ s.t. $r_{i}\in\var(t)$,
then $s{\gg}$ is included into $\textnormal{strands}_{St}(t)$;
or 
\item
for each strand $s$ in $St$ of the form
$::r_1,\ldots,r_k::\;[m_1^\pm,\ldots \mid \ldots,m_n^\pm]$,
if there is another strand $s'$
of the form 
$::r'_1,\ldots,r'_{k'}::\;[w_1^\pm,\ldots \mid \ldots,w_{n'}^\pm]$
in $\textnormal{strands}_{St}(t)$,
and
there are $i\in\{1,\ldots,k\}$ and $j\in\{1,\ldots,n'\}$
s.t. $r_{i}\in\var(w_{j})$,
then $s{\gg}$ is included into $\textnormal{strands}_{St}(t)$.
\end{itemize}
\end{definition}

Now, we formally define how to resuscitate a state.

\begin{definition}[Resuscitation]\label{def:resuscitation}
Given a topmost rewrite theory $\caR_\caP = \linebreak[4](\Symbols_\caP,E_{\caP},\widetilde{R_{\caP}})$
representing protocol $\caP$
and a state $St$ containing an intruder fact $\inI{t}$ such that $t$ 
is a super-lazy term, i.e., $St = ss \,\&\, (\inI{t},ik)$
where $ss$ is a term denoting a set of strands
and $ik$ is a term denoting the rest of the intruder knowledge.
Let $\overline{St}$ be the ghost version of $St$.
Let $St'$ be a state such that
$\overline{St} \narrto^{*}_{\sigma,\widetilde{R_{\caP}}^{-1},E_\caP} St'$
and
$\sigma(t)$ is not a super-lazy term.
Let $\sigma_t=\sigma\restrict{\var(t)}$.
The \emph{reactivated} (or \emph{resuscitated}) version of $St$ 
w.r.t. state $St'$ and substitution $\sigma_t$ 
is defined as 
$\widehat{St}=\sigma_t(ss) \,\&\, \sigma_t(ik) \,\&\, \textnormal{strands}_{St'}(\sigma_t(t))$.
\end{definition}

Let us now prove the completeness of this state space reduction technique.

\begin{theorem}
Given a topmost rewrite theory $\caR_\caP = (\Symbols_\caP,E_{\caP},\widetilde{R_{\caP}})$
representing protocol $\caP$
and a state $St$ containing an intruder fact $\inI{t}$ such that $t$ 
is a super-lazy term,
if there exist an initial state $St_\mathit{ini}$ 
and substitution $\theta$
such that
$St 
\narrto^{*}_{\theta,\widetilde{R_{\caP}}^{-1},E_\caP} St_\mathit{ini}$,
then
(i)
there exist a state $St'$
and substitutions $\tau,\tau'$
such that
$St 
\narrto^{*}_{\tau,\widetilde{R_{\caP}}^{-1},E_\caP} St'$,
$\theta \congr{E_\caP} \tau\circ\tau'$,
and
$\tau(t)$ is not a super-lazy term,
and
(ii)
there exist a reactivated version $\widehat{St}$ of $St$ w.r.t. $St'$ and $\tau$,
an initial state $St'_\mathit{ini}$,
and substitutions $\theta'$, $\rho$
such that
$\widehat{St} 
\narrto^{*}_{\theta',\widetilde{R_{\caP}}^{-1},E_\caP} St'_\mathit{ini}$,
$\theta \congr{E_\caP} \theta'\circ\rho$,
and
$\rho(St'_\mathit{ini}) \subseteq_{E_\caP} St_\mathit{ini}$.
\end{theorem}
\begin{proof}
%
The sequence from $St$ to 
$St_\mathit{ini}$ can be decomposed into two fragments,
computing substitutions $\tau$, $\tau'$, respectively,
such that 
$\tau$ is the smallest part of $\theta$
that makes 
$\tau(t)$ not a super-lazy term.
That is,
there is a state $St'$ and substitutions $\tau$, $\tau'$
such that
$\tau(t)$ is not a super-lazy term,
$\theta=\tau\circ\tau'$,
$St 
\narrto^{*}_{\tau,\widetilde{R_{\caP}}^{-1},E_\caP} St'
\narrto^{*}_{\tau',\widetilde{R_{\caP}}^{-1},E_\caP} St_\mathit{ini}$,
and
the sequence
$St 
\narrto^{*}_{\tau,\widetilde{R_{\caP}}^{-1},E_\caP} St'$
can be viewed as
$St = St_0 
\narrto_{\tau_1,\widetilde{R_{\caP}}^{-1},E_\caP}
\cdots
\narrto_{\tau_k,\widetilde{R_{\caP}}^{-1},E_\caP} St_k = St'$
such that for all $i \in \{1,\ldots,k-1\}$,
$\tau_i(t)$ is a super-lazy term.
However, using the completeness results of narrowing,
Theorem~\ref{thm:hosc06},
there must be a narrowing sequence from $\overline{St}$ computing
such substitution $\tau$.
That is,
there is a state $St''$
such that
$\overline{St} 
\narrto^{*}_{\tau,\widetilde{R_{\caP}}^{-1},E_\caP} St''$
and $St''$ differs from $St'$ 
(modulo $E_\caP$-equivalence and variable renaming)
only in that $\inI{\tau(t)}$
is replaced by $ghost(\tau(t))$.
Let 
$\tau_t=\tau\restrict{\var(t)}$,
there exists a substitution $\tau''$ s.t.
$\tau \congr{E_\caP} \tau_t\circ\tau''$.
Let $\widehat{St}$ be the resuscitated version of $St$ w.r.t. state $St''$ and substitution $\tau_t$.
Then,
by narrowing completeness,
i.e., Theorem~\ref{thm:hosc06},
there exist a state $St'_\mathit{ini}$
and substitutions $\sigma,\rho$ such that
$\widehat{St} 
\narrto^{*}_{\sigma,\widetilde{R_{\caP}}^{-1},E_\caP} St'_\mathit{ini}$,
$\tau'' \circ \tau' \congr{E_\caP} \sigma\circ\rho$,
and 
$\rho(St'_\mathit{ini}) \congr{E_\caP} St_\mathit{ini}$.
\qed
\end{proof}

\subsubsection{Improving the Super-Lazy Intruder.}

When we detect a state $St$ with a super lazy term $t$,
we may want to analyze whether the variables of $t$ may be eventually
instantiated or not
before creating a ghost state. 
The following definition provides the key idea.

\begin{definition}[Void Super-Lazy Term]
Given a topmost rewrite theory 
	\linebreak
$\caR_\caP = (\Symbols_\caP,E_{\caP},\widetilde{R_{\caP}})$
representing protocol $\caP$,
and a state $St$ containing an intruder fact $\inI{t}$ such that $t$ 
is a super-lazy term,
if 
for each strand
$[m_1^\pm, \ldots, m_{j-1}^\pm \mid m_{j}^\pm, \ldots, m_k^\pm]$ in $St$
and each $i\in\{1,\ldots,j-1\}$,
$\var(t)\cap\var(m_{i})=\emptyset$,
and
for each term $\inI{w}$ in the intruder's knowledge, 
$\var(t)\cap\var(w)=\emptyset$,
then,
$t$ is called a \emph{void super-lazy term}.
\end{definition}

\begin{proposition}\label{prp:superlazy-improved}
Given a topmost rewrite theory 
$\caR_\caP = (\Symbols_\caP,E_{\caP},\widetilde{R_{\caP}})$
representing protocol $\caP$
and a state $St$ containing an intruder fact $\inI{t}$ such that $t$ 
is a void super-lazy term,
let $\overline{St}$ be the ghost version of $St$ w.r.t. 
the void super-lazy term $t$.
If 
there exist an initial state $St_\mathit{ini}$ 
and a substitution $\theta$
such that
$St 
\narrto^{*}_{\theta,\widetilde{R_{\caP}}^{-1},E_\caP} St_\mathit{ini}$,
then
there exist an initial state $St'_\mathit{ini}$
and 
substitutions $\sigma,\rho$
such that 
$\overline{St} 
\narrto^{*}_{\sigma,\widetilde{R_{\caP}}^{-1},E_\caP} St'_\mathit{ini}$,
$\theta \congr{E_\caP} \sigma \circ \rho$, and
$\rho(St'_\mathit{ini}) \subseteq_{E_\caP} St_\mathit{ini}$.
\end{proposition}
\begin{proof}
Since $t$ is a super-lazy term, 
$St_\mathit{ini}$ contains a sequence of 
intruder strands of $\caS_\caP$ generating $t$.
Let $\theta_t=\theta\restrict{\var(t)}$,
there exists a substitution $\theta'$ s.t.
$\theta \congr{E_\caP} \theta_t\circ\theta'$.
Since $t$ is a void super-lazy term, 
there is a state $St''_\mathit{ini}$
such that
$\theta'(\overline{St}) \rewrites{\widetilde{R_{\caP}}^{-1},E_\caP}
St''_\mathit{ini}$.
Then, by narrowing completeness, i.e., Theorem~\ref{thm:hosc06},
there are an initial state $St'_\mathit{ini}$
and 
substitutions $\sigma,\rho$
such that 
$\overline{St} 
\narrto^{*}_{\sigma,\widetilde{R_{\caP}}^{-1},E_\caP} St'_\mathit{ini}$,
$\theta' \congr{E_\caP} \sigma \circ \rho$, and
$\rho(St'_\mathit{ini}) \subseteq_{E_\caP} St''_\mathit{ini}$.
Finally, $St''_\mathit{ini} \subseteq_{E_\caP} St_\mathit{ini}$,
since $St_\mathit{ini}$ simply has the strands generating $t$
that $St''_\mathit{ini}$ does not contain.
\qed
\end{proof}

\subsubsection{Interaction with Transition Subsumption.}\label{sec:interaction}

When a ghost state is reactivated, we see from the above
definition that such a reactivated state will be $\caP$-subsumed by 
the original state that raised the ghost expression. 
Therefore, the transition subsumption relation
$\blacktriangleright$ of 
Section~\ref{sec:folding} has to be slightly modified to avoid checking 
a resuscitated state against its predecessor ghost state. 
Now, let us formally state this problem.

\begin{definition}[Resuscitated Child]\label{def:resuscitated-child}
Given a topmost rewrite theory $\caR_\caP = (\Symbols_\caP,E_{\caP},\widetilde{R_{\caP}})$
representing protocol $\caP$
and two non-initial states $St$ and $St'$
such that $St$ contains an intruder fact $\inI{t}$
and $t$ is a super-lazy term, 
we 
say $St'$ is a \emph{resuscitated child} of $St$,
written $St \curvearrowright St'$,
if:
\begin{enumerate}

\item
given the ghost version $\overline{St}$  of $St$ w.r.t. 
the super-lazy term $t$,
then
there exist states $St_1,\ldots,St_k$,
substitutions $\tau_1,\ldots,\tau_k$,
and $i\in\{1,\ldots,k\}$
such that
$$
\overline{St} 
\narrto_{\tau_1,\widetilde{R_{\caP}}^{-1},E_\caP} St_1
\cdots
St_{i-1}
\narrto_{\tau_i,\widetilde{R_{\caP}}^{-1},E_\caP} St_i
\cdots
St_{k-1}
\narrto_{\tau_k,\widetilde{R_{\caP}}^{-1},E_\caP} St_k,
$$
$\tau_j(t)$ is a super-lazy term for $1\leq j\leq i-1$, 
and
$\tau_i(t)$ is not a super-lazy term,
and 

\item
given the reactivated version $\widetilde{St}$ of $St$ w.r.t. $St_i$ and $\tau=\tau_1\circ\cdots\circ\tau_i$
and
$\tau_t=\tau\restrict{\var(t)}$,
there exist substitutions $\tau'_1,\ldots,\tau'_k$ such that
$\tau_j=\tau_t\circ\tau'_j$ for $1\leq j\leq k$,
states $St'_1,\ldots,St'_k$,
and a narrowing sequence
$$
\widetilde{St} 
\narrto_{\tau'_1,\widetilde{R_{\caP}}^{-1},E_\caP} St'_1
\cdots
St'_{k-1}
\narrto_{\tau'_k,\widetilde{R_{\caP}}^{-1},E_\caP} St'_k
$$
\item
then
there is $j\in\{1,\ldots,k\}$ such that
$St' \congr{E_\caP} St'_j$.
\end{enumerate}
\end{definition}

\begin{proposition}\label{prp:problem-interaction}
Given a topmost rewrite theory $\caR_\caP = (\Symbols_\caP,E_{\caP},\widetilde{R_{\caP}})$
representing protocol $\caP$
and two non-initial states $St$ and $St'$
such that $St$ contains an intruder fact $\inI{t}$
and $t$ is a super-lazy term, 
if $St \curvearrowright St'$,
then 
$St \blacktriangleright St'$
and reachability completeness is lost.
\end{proposition} 
\begin{proof}
Since $\overline{St}$ is similar to $St$ but $\inI{t}$ has been replaced by $ghost(t)$, and $\widehat{St}$ contains all the strands and positive intruder facts of $St$ but instantiated with $\tau\restrict{\var(t)}$,
then
for the sequences
$$
\overline{St} 
\narrto_{\tau_1,\widetilde{R_{\caP}}^{-1},E_\caP} St_1
\cdots
St_{k-1}
\narrto_{\tau_k,\widetilde{R_{\caP}}^{-1},E_\caP} St_k
$$
and
$$
\widetilde{St} 
\narrto_{\tau'_1,\widetilde{R_{\caP}}^{-1},E_\caP} St'_1
\cdots
St'_{k-1}
\narrto_{\tau'_k,\widetilde{R_{\caP}}^{-1},E_\caP} St'_k
$$
we have that
$St_j \blacktriangleright St'_j$
for $j\in\{1,\ldots,k\}$,
since $St'_j$
contains all the strands and positive intruder facts of $St_j$ but instantiated with $\tau\restrict{\var(t)}$.
Reachability completeness is lost because 
if there is an initial state $St_\mathit{ini}$
and substitution $\tau'$
such that
$St 
\narrto^{*}_{\tau,\widetilde{R_{\caP}}^{-1},E_\caP} St'
\narrto^{*}_{\tau',\widetilde{R_{\caP}}^{-1},E_\caP} St_\mathit{ini}$,
then,
since 
$St$ is replaced 
by $\overline{St}$ during the backwards reachability analysis and later 
$\overline{St}$ is replaced by $\widehat{St}$,  
when Maude-NPA finds that $St_j \blacktriangleright St'_j$,
it removes $St'_j$ from the backwards reachability analysis,
(possibly) leaving no successor of $St$ leading to $St_\mathit{ini}$.
\qed
\end{proof}

The simplest way of ensuring whether or not $St_1 \curvearrowright St_2$ is to examine the relative positions of $St_1$ and $St_2$ in the search tree
as well as the narrowing steps between them in the form established by Definition~\ref{def:resuscitated-child}.  However, for reasons of efficiency, we want to
keep examinations of the search tree to a minimum, and restrict ourselves as much as possible to looking at information in the state itself.  Thus, we make
use of information that is already in the state, the message sequence first mentioned in Section 3.1.  We find, that after making minor modifications to this
message sequence to take account of resuscitated ghosts, a simple syntactic check on the sequence can provide a relation that approximates $\curvearrowright$.

In order to formally identify when a resuscitated state
must not be erroneously discarded by $\blacktriangleright$,
we extend protocol states to have 
the actual message exchange sequence between principal or intruder strands
and 
add a new expression
$resuscitated(m)$
to indicate when a state has been resuscitated.
The actual set of rewrite rules extended 
to compute the exchange sequence is as follows,
where $X$ is a variable denoting an exchange sequence:

\noindent
\begin{scriptsize}%
\begin{align}%
[L ~|~ M^-, L']\  \&\ SS\ \&\ (\inI{M},IK)
  \ \& \ (M^-,X)
  &\to 
  [L, M^- ~|~ L']\  \&\ SS\ \&\ (\inI{M},IK)
  \ \& \ X
  \notag
  \\
 [L ~|~ M^+, L']\  \&\ SS\ \&\ IK 
 \hspace{8.6ex}
  \ \& \ (M^+,X)
  &\to 
  [L, M^+ ~|~ L']\  \&\ SS\ \&\ IK
  \ \& \ X
  \notag
  \\
 [L ~|~ M^+, L']\  \&\ SS\ \&\ (\nI{M},IK) 
  \ \& \ (M^+,X)
  &\to 
  [L, M^+ ~|~ L']\  \&\ SS\ \&\ (\inI{M},IK)
  \ \& \ X
  \notag
\end{align}%
\vspace{-6ex}
\begin{align}%
\mbox{for each }[~ l_1,\ u^+,\ l_2~] \in \caS_{\caP}:
[~ l_1 ~|~ u^+, l_2~] \, \&\, SS\, \,\&\, (\nI{u},IK)
  \ \& \ (u^+,X)
\to
SS \,\&\, (\inI{u},IK)
  \ \& \ X
\notag
\end{align}%
\end{scriptsize}%
Completeness reachability is obviously preserved for this set of rules
and for the obvious extensions to $\overline{R_{\caP}}$ and
$\widetilde{R_\caP}$.
For instance, the resuscitated state of Example~\ref{ex:resuscitated} 
will be written as follows,
where the resuscitated message is the first item in the exchange sequence:

{\small
$$
\begin{array}{l}
[\ nil \mid exp(X,n(a,r))^-,\ e(exp(X,n(a,r), sec(a, r'')))^-,\ sec(a, r'')^+\ ]\ \& \\[1ex]
:: r,r'' ::\\[1ex]
[\ (a ; B' ; exp(g,n(a,r)))^+,\ (B' ; a ; X)^-,\  (e(exp(X,n(a,r)), sec(a,r'')))^+
\mid nil \ ]\ \&\\[1ex]
::r':: 
[\ (A ; B ; Y)^-, (B ; A ; exp(g,n(B,r')))^+ \mid 
(e(exp(Y,n(B,r')),sec(a,r'')))^-\ ]\ \&\\[1ex]
(\inI{e(exp(Y,n(B,r')),sec(a,r''))},\ 
\inI{exp(X,n(a,r))},\\[1ex]
\ \inI{e(exp(X,n(a,r)), sec(a, r'')))},\ 
\nI{sec(a,r'')})\ \&\\[1ex]
(resuscitated(exp(X, n(a, r))),\  
exp(X, n(a, r)))^-,\  
e(exp(X, n(a, r)), sec(a, r'')))^-,\\[1ex]  
\ (sec(a, r''))^+,\  
(exp(Y, n(b, r')))^-,\  
(sec(a, r''))^-,\  
(e(exp(Y, n(b, r')), sec(a, r'')))^+,\\[1ex]  
\ (e(exp(Y, n(b, r')), sec(a, r'')))^-
)
\end{array}
$$
}

In \cite{EscobarMeadowsMeseguerESORICS08}, we provided a very simple
rule 
for approximating Definition~\ref{def:resuscitated-child}.

\begin{definition}
Given a topmost rewrite theory $\caR_\caP = (\Symbols_\caP,E_{\caP},\widetilde{R_{\caP}})$
representing protocol $\caP$
and two non-initial states $St_{1}, St_{2}$, 
we write $St_{1} \dashrightarrow St_{2}$
if 
either 
$St_{1}$ does not contain an expression $\textnormal{ghost}(m)$
for a message term $m$
or
$St_{1}$ does contain an expression $\textnormal{ghost}(m)$
for a message term $m$
but
$St_2$ does not contain the expression $resuscitated(m)$.
\end{definition}

The following result establishes that $\dashrightarrow$ is an approximation of $\curvearrowright$. The proof is straightforward.

\begin{lemma}\label{lem:lazy-1}
Given a topmost rewrite theory $\caR_\caP = (\Symbols_\caP,E_{\caP},\widetilde{R_{\caP}})$
representing protocol $\caP$
and two non-initial states $St_{1}, St_{2}$, 
if $St_1 \curvearrowright St_2$,
then $St_1 \dashrightarrow St_2$.
\end{lemma}

Now, we can provide a better transition subsumption relation.

\begin{definition}[$\caP$-subsumption relation II]\label{def:lazy-1}
Given a topmost rewrite theory $\caR_\caP = (\Symbols_\caP,E_{\caP},\widetilde{R_{\caP}})$
representing protocol $\caP$
and two non-initial states $St_{1}, St_{2}$, 
we write $St_{1} \blacktriangleright_\textit{II} St_{2}$
and say that $St_{2}$ is \emph{$\caP$-subsumed} by $St_{1}$
if there is a substitution $\theta$ s.t. 
$\theta(St_{1}) \triangleright St_{2}$
and
$\theta(St_{1}) \not\dashrightarrow St_2$.
\end{definition}

\noindent
Reachability completeness is straightforward from
Lemma~\ref{lem:lazy-1}
and
Proposition~\ref{prp:problem-interaction},
since 
$St_{1} \not\dashrightarrow St_2$
implies
$St_{1} \not\curvearrowright St_2$.

Though this method solves the problem, it disables almost completely  the transition subsumption for those states after a resuscitation,
since $\dashrightarrow$ is a bad approximation of $\curvearrowright$.
Here, we provide a more concise definition of the interaction
between the transition subsumption and the super-lazy intruder
reduction techniques.

We characterize those states after a resuscitation
that are truly linked to the parent state.
First, we identify those states that are directly resuscitated versions
of a former state.
Intuitively, by comparing the exchange sequences of the two states,
we can see whether
the exchange sequence of the former 
is $(L_1,L_2,M_1^-,L_3)$
and it has a ghost expression $\textnormal{ghost}(M_1)$,
whereas
the exchange sequence of the resuscitated version
is $(L_1,resuscitated(M_1),L_2,M_1^-,L_3)$.

\begin{definition}
Given a topmost rewrite theory $\caR_\caP = (\Symbols_\caP,E_{\caP},\widetilde{R_{\caP}})$
representing protocol $\caP$
and two non-initial states $St_{1}, St_{2}$, 
we say that $St_2$ is a \emph{direct resuscitated version} of $St_1$,
written $S_1 \twoheadrightarrow St_2$,
if 
there are messages $M_1$ and $M_2$ 
and
a substitution $\rho$ such that
\begin{enumerate}
\item
state $St_1$ has a ghost of the form $ghost(M_1)$,
\item
the exchange sequence of state $St_1$ is of the form
$$(L_1,L_2,M_1^-,L_3)$$
\item
the exchange sequence of state $St_2$ is of the form
$$(L'_1,resuscitated(M_2),L'_2,M_2^-,L'_3),$$
\item
and 
$\rho(L_1,L_2,M_1^-,L_3) \congr{E_\caP} (L'_1,L'_2,M_2^-,L'_3).$
\end{enumerate}
\end{definition}

\noindent
Relation $\twoheadrightarrow$ is closer to $\curvearrowright$.

\begin{lemma}\label{lem:lazy-2}
Given a topmost rewrite theory $\caR_\caP = (\Symbols_\caP,E_{\caP},\widetilde{R_{\caP}})$
representing protocol $\caP$
and two non-initial states $St_{1}, St_{2}$, 
if $St_1 \twoheadrightarrow St_2$,
then $St_1 \curvearrowright St_2$.
\end{lemma}

%

However, 
$St_1 \curvearrowright St_2$
does not imply
$St_1 \twoheadrightarrow St_2$
and
we have to go even further.
Relation $St_1 \twoheadrightarrow St_2$ takes into account only 
whether $St_2$ is a resuscitated version of $St_1$, but
does not consider what happens beyond the state that produced the instantiation that reactivated the ghost state.
Intuitively, now we compare the exchange sequences of the two states
to see whether
the exchange sequence of the first 
is $(L_1,L_2,L_3,M_1^-,L_4)$
and it has a ghost expression $\textnormal{ghost}(M_1)$,
whereas
the exchange sequence of the second
is $(L_1,M_1^+,L_2,resuscitated(M_1),L_3,M_1^-,L_4)$.
Indeed, a recursive definition can be given here that becomes extremely useful when several resuscitations have happened in a concrete state.

\begin{definition}
Given a topmost rewrite theory $\caR_\caP = (\Symbols_\caP,E_{\caP},\widetilde{R_{\caP}})$
representing protocol $\caP$
and two non-initial states $St_{1}, St_{2}$, 
we say that $St_2$ is a \emph{resuscitated version} of $St_1$,
written $S_1 \twoheadrightarrow^+ St_2$,
if 
$S_1 \twoheadrightarrow St_2$
or
there are messages $M_1$ and $M_2$, 
a substitution $\rho$,
and sequences $L'_1,L''_1$
 such that:
\begin{enumerate}
\item
state $St_1$ has a ghost of the form $ghost(M_1)$,
\item
the exchange sequence of state $St_1$ is of the form
$$(L_1,L_2,L_3,M_1^-,L_4)$$
\item 
the exchange sequence of state $St_2$ is of the form
$$(L'_1,L''_1,M_2^+,L'_2,resuscitated(M_2),L'_3,M_2^-,L'_4)$$
\item
$\rho(L_2,L_3,M_1^-,L_4) \congr{E_\caP} (L'_2,L'_3,M_2^-,L'_4)$
\item
$L''_1$ is the longest sequence such that
each message $m^\pm$ in $L''_1$ has message $\rho(M_1)$ as a subterm
\item
and
either
\begin{enumerate}
\item
$\rho(L_1) \congr{E_\caP} L'_1$
or
\item
$St'_{1} \twoheadrightarrow^+ St'_2$ 
where 
$St'_1$ is $St_1$ without the $ghost(M_1)$ expression
and
$St'_2$ is $St_2$ with the shorter 
exchange sequence 
$(L_1',L'_2,L'_3,M_2^-,L'_4)$.
\end{enumerate}
\end{enumerate}
\end{definition}

The following result establishes that $\twoheadrightarrow^+$ is a better approximation of $\curvearrowright$
than $\dashrightarrow$. The proof is straightforward.

\begin{lemma}\label{lem:lazy-3}
Given a topmost rewrite theory $\caR_\caP = (\Symbols_\caP,E_{\caP},\widetilde{R_{\caP}})$
representing protocol $\caP$
and two non-initial states $St_{1}, St_{2}$, 
if $St_1 \curvearrowright St_2$,
then $St_1 \twoheadrightarrow^+ St_2$.
\end{lemma}

Now, we can provide a better transition subsumption relation.
\begin{definition}[$\caP$-subsumption relation III]\label{def:lazy-3}
Given a topmost rewrite theory $\caR_\caP = (\Symbols_\caP,E_{\caP},\widetilde{R_{\caP}})$
representing protocol $\caP$
and two non-initial states $St_{1}, St_{2}$, 
we write $St_{1} \blacktriangleright_\textit{III} St_{2}$
and say that $St_{2}$ is \emph{$\caP$-subsumed} by $St_{1}$
if there is a substitution $\theta$ s.t. 
$\theta(St_{1}) \triangleright St_{2}$
and $St_1 \not\twoheadrightarrow^+ St_2$.
\end{definition}

Finally, reachability completeness is straightforward from
Lemma~\ref{lem:lazy-3}
and
Proposition~\ref{prp:problem-interaction},
since 
$St_{1} \not\twoheadrightarrow^+ St_2$
implies
$St_{1} \not\curvearrowright St_2$.

\section{Experimental Evaluation}\label{sec:experiments}

%
%
%
%
%
%
%
%
%
%
%
%
%
%
%

\begin{table}[t]
{
\scriptsize

\centering

\begin{tabular}{@{}|c|rrrrr|rrrrr|c|@{}}

\hline

Protocol 
& \multicolumn{5}{c|}{none}
& \multicolumn{5}{c|}{Grammars} & \%
\\
\hline

{\sf NSPK}
& 5 &19 &136 &642 &4021
& 4 &12 &49 &185 &758
& 81
\\
\hline

{\sf NSL}
& 5 &19 &136 &642 &4019
& 4 &12 &50 &190 &804
& 79
\\
\hline

{\sf SecReT06} 
& 1 &6 &22 &119 &346
& 1 &2 &6 &15 &36
& 89
\\
\hline

{\sf SecReT07}
& 6 &20 &140 &635 &4854
& 6 &17 &111 &493 &3823
& 21
\\
\hline

{\sf DH}
& 1 &14 &38 &151 &816
& 1 &6 &14 &37 &105
& 87
\\

\hline

\end{tabular}


}

\vspace{1ex}

{
\scriptsize

\centering

\begin{tabular}{@{}|c|rrrrr|rrrrr|c|@{}}

\hline

Protocol 
& \multicolumn{5}{c|}{none}
& \multicolumn{5}{c|}{Input First} & \%
\\
\hline

{\sf NSPK}
& 5 &19 &136 &642 &4021
& 11 &123 &1669 &26432 & N/A
& 0
\\
\hline

{\sf NSL}
& 5 &19 &136 &642 &4019
& 11 &123 &1666 &26291 &N/A
& 0
\\
\hline

{\sf SecReT06} 
& 1 &6 &22 &119 &346
& 11 &133 &1977 &32098 &N/A
& 0
\\
\hline

{\sf SecReT07}
& 6 &20 &140 &635 &4854
& 11 &127 &3402 &N/A &N/A
& 0
\\
\hline

{\sf DH}
& 1 &14 &38 &151 &816
& 14 &135 &1991 &44157 &N/A
& 0
\\

\hline

\end{tabular}


}

\vspace{1ex}

{
\scriptsize

\centering

\begin{tabular}{@{}|c|rrrrr|rrrrr|c|@{}}

\hline

Protocol 
& \multicolumn{5}{c|}{none}
& \multicolumn{5}{c|}{Inconsistency} & \%
\\
\hline

{\sf NSPK}
& 5 &19 &136 &642 &4021
& 5 &18 &95 &310 &650
& 83
\\
\hline

{\sf NSL}
& 5 &19 &136 &642 &4019
& 5 &18 &95 &310 &650
& 83
\\
\hline

{\sf SecReT06} 
& 1 &6 &22 &119 &346
& 1 &6 &22 &114 &326
& 5
\\
\hline

{\sf SecReT07}
& 6 &20 &140 &635 &4854
& 6 &18 &107 &439 &3335
& 31
\\
\hline

{\sf DH}
& 1 &14 &38 &151 &816
& 1 &12 &12 &56 &128
& 84
\\

\hline

\end{tabular}


}

\vspace{1ex}

{
\scriptsize

\centering

\begin{tabular}{@{}|c|rrrrr|rrrrr|c|@{}}

\hline

Protocol 
& \multicolumn{5}{c|}{none}
& \multicolumn{5}{c|}{Transition Subsumption} & \%
\\
\hline

{\sf NSPK}
& 5 &19 &136 &642 &4021
& 5 &15 &61 &107 &237
& 94
\\
\hline

{\sf NSL}
& 5 &19 &136 &642 &4019
& 5 &15 &61 &107 &237
& 94
\\
\hline

{\sf SecReT06} 
& 1 &6 &22 &119 &346
& 1 &6 &15 &39 &78
& 77
\\
\hline

{\sf SecReT07}
& 6 &20 &140 &635 &4854
& 6 &15 &61 &165 &506
& 89
\\
\hline

{\sf DH}
& 1 &14 &38 &151 &816
& 1 &14 &26 &102 &291
& 64
\\

\hline

\end{tabular}


}

\vspace{1ex}

{
\scriptsize

\centering

\begin{tabular}{@{}|c|rrrrr|rrrrr|c|@{}}

\hline

Protocol 
& \multicolumn{5}{c|}{none}
& \multicolumn{5}{c|}{Super-lazy Intruder} & \%
\\
\hline

{\sf NSPK}
& 5 &19 &136 &642 &4021
& 5 &19 &136 &641 &3951
& 1
\\
\hline

{\sf NSL}
& 5 &19 &136 &642 &4019
& 5 &19 &136 &641 &3949
& 2
\\
\hline

{\sf SecReT06} 
& 1 &6 &22 &119 &346
& 1 &6 &22 &119 &340
& 2
\\
\hline

{\sf SecReT07}
& 6 &20 &140 &635 &4854
& 6 &16 &44 &134 &424
& 91
\\
\hline

{\sf DH}
& 1 &14 &38 &151 &816
& 1 &14 &38 &138 &525
& 35
\\

\hline

\end{tabular}


}

\vspace{1ex}

{
\scriptsize

\centering

\begin{tabular}{@{}|c|rrrrr|rrrrr|c|@{}}

\hline

Protocol 
& \multicolumn{5}{c|}{none}
& \multicolumn{5}{c|}{All optimizations} & \%
\\
\hline

{\sf NSPK}
& 5 &19 &136 &642 &4021
& 4 &6 &4 &2 &1
& 99
\\
\hline

{\sf NSL}
& 5 &19 &136 &642 &4019
& 4 &7 &6 &2 &0
& 99
\\
\hline

{\sf SecReT06} 
& 1 &6 &22 &119 &346
& 2 &3 &2 &- &-
& 99
\\
\hline

{\sf SecReT07}
& 6 &20 &140 &635 &4854
& 5 &1 &1 &1 &-
& 99
\\
\hline

{\sf DH}
& 1 &14 &38 &151 &816
& 4 &6 &10 &9 &12
& 99
\\

\hline

\end{tabular}

\caption{Number of states for 1,2,3, and 4 backwards narrowing steps
comparing each optimization of Sections
\ref{sec:grammars},\ref{sec:input-first},\ref{sec:inconsistent},\ref{sec:folding}, and \ref{sec:super-lazy-intruder}.}

\label{table-all}
}
\end{table}

\begin{table}[t]
{
\scriptsize

\centering

\begin{tabular}{@{}|c|p{5.5cm}|@{}}

\hline
Protocol & Finite State Space Achieved by:\\
\hline
{\sf NSPK} & Grammars and Subsumption\\
\hline
{\sf NSL} & Grammars and Subsumption\\
\hline
{\sf SecReT06} & Subsumption or (Grammars and Lazy)\\
\hline
{\sf SecReT07} & Subsumption and Lazy\\
\hline
{\sf DH} & Grammars and Subsumption\\
\hline

\end{tabular}
\caption{Finite state space achieved by reduction techniques}
\label{tableFinite}
}
\end{table}

In Table 
\ref{table-all},
we summarize the experimental evaluation of
the impact of the different state space reduction techniques 
for various example protocols searching up to depth $4$.
We measure several numerical values for the techniques:
(i) number of states at each backwards narrowing step,
and
(ii) whether the state space is finite or not.
The experiments have been performed on a MacBook with 2 Gb RAM
using Maude 2.6.
All protocol specifications are included in the official Maude-NPA distribution\footnote{Available at \url{http://maude.cs.uiuc.edu/tools/Maude-NPA}.}.
The protocols are the following:
(i) NSPK, the standard Needham-Schroeder protocol,
(ii) NSL, the standard Needham-Schroeder protocol with Lowe's fix
(which is secure and our tool can prove it),
(iii) SecReT06, a protocol with an attack using type confusion
and a bounded version of associativity that
we presented in \cite{escobar-meadows-meseguer-secret06},
(iv) SecReT07, a short version of the Diffie-Hellman protocol 
that
we presented in \cite{escobar-hendrix-meadows-meseguer-secret07},
and
(v) DH, the Diffie-Hellman protocol of Example~\ref{ExDH}.
Note that the label ``-'' means that the reachability analysis finished some levels before and the label ``N/A'' means that the execution was stopped after a reasonably large execution time.

The overall percentage of state-space reduction for each protocol
and an average ($99\%$) suggest that our combined techniques are remarkably
effective (the reduced number of states is on average
only $1\%$ or less of the original 
number of states).
The state reduction achieved by
consuming input messages first is difficult to analyze,
since the reduction shown in 
Table~\ref{table-all} for this optimization (labelled as ``Input First'')
is $0$.
The reason is that 
it can reduce the number of states in protocols 
that contain several input messages in the strands,  as in the NSPK protocol, 
but in general it simply reduces the length of the narrowing sequences
and therefore more states can be generated at an earlier depth of the narrowing tree compared to the case where the optimization is not used.
Table \ref{tableFinite} summarizes the different techniques
yielding a finite space for each protocol. 
The use of grammars and the transition subsumption are clearly the
most useful techniques in general.
Indeed, all examples have a \emph{finite search space} 
thanks to the combined use of the different state space reduction techniques.
Note that grammars are insufficient 
to obtain a finite space for
the SecReT07 example, while subsumption and the super lazy intruder are essential in this case.

\section{Concluding Remarks}\label{sec:conclusions}

The Maude-NPA can analyze the security of cryptographic protocols,
modulo given algebraic properties of the protocol's cryptographic
functions in executions with an unbounded number of sessions and with
no approximations or data abstractions.  In this full generality,
protocol security properties are well-known to be undecidable.  The
Maude-NPA uses backwards narrowing-based search from a symbolic
description of a set of attack states by means of patterns to try to reach an
initial state of the protocol.  If  an attack state
is reachable from an initial state, the Maude-NPA's complete narrowing 
methods are guaranteed to prove it.  But if the protocol is
secure, the backwards search may be infinite and never terminate.

It is therefore very important, both for efficiency and to achieve full
verification whenever possible when a protocol is secure, to use
\emph{state-space reduction techniques} that: (i) can drastically cut down
the number of states to be explored; and (ii) have in practice a good
chance to make the, generally infinite, search space finite without
compromising the completeness of the analysis; that is, so that if a protocol is indeed
secure, failure to find an attack in such a finite state space guarantees
the protocol's security 
for that attack 
relative to the assumptions about the intruder actions and the algebraic properties.
We have presented a
number of state-space reduction techniques used in combination by the
Maude-NPA for exactly these purposes.
We have given precise characterizations of
theses techniques and have shown that they preserve completeness, so that
if no attack is found and the state space is finite, full verification
of the given security property is achieved.

Using several representative examples we have also given an
experimental evaluation of these  techniques.
Our experiments support the conclusion that, when used in combination,
these techniques: (i) typically provide drastic state space reductions;
and (ii) they can often yield a \emph{finite} state space, so that
whether the desired security property holds or not can in fact be
decided automatically, in spite of the general undecidability of such problems.

\section*{Acknowledgements}
We are very thankful to Sonia Santiago for her help
on evaluating the different benchmarks.

\bibliography{tex,unification}
\bibliographystyle{elsarticle-harv}

\end{document}